\documentclass[12pt]{article}
\usepackage[letterpaper, margin=1in]{geometry}
\usepackage{amsmath, amsthm, amssymb,amsfonts, graphicx, subfigure,cases}
\usepackage[colorlinks,citecolor=blue,linkcolor=blue, urlcolor=blue, anchorcolor=blue]{hyperref}
\usepackage{braket}
\usepackage[ruled,linesnumbered]{algorithm2e}
\usepackage{enumerate}
\usepackage{booktabs}
\usepackage{tablefootnote}
\usepackage{tikz}
\usetikzlibrary{shapes.geometric}

\newtheorem{lemma}{Lemma}
\newtheorem{theorem}{Theorem}

\newtheorem{fact}{Fact}
\newtheorem{claim}{Claim}
\newtheorem{problem}{Problem}

\theoremstyle{definition}

\newtheorem{remark}{Remark}

\def\norm#1{\left\| #1 \right\|}
\def\proj#1{\ket{#1}\bra{#1}}

\usepackage{cleveref}
\crefname{equation}{Eq.}{Eqs.}
\Crefname{equation}{Eq.}{Eqs.}

\graphicspath{{./figures/}}

\title{Quantum phase discrimination with applications to quantum search on graphs\footnote{Authors listed in alphabetical order}}
\author{Guanzhong Li\thanks{Email: ligzh9@mail2.sysu.edu.cn},~ Lvzhou Li\thanks{Email: lilvzh@mail.sysu.edu.cn (corresponding author)}~ and Jingquan Luo\thanks{Email: luojq25@mail2.sysu.edu.cn}\\
\small{{\it Institute of Quantum Computing and Software,}}
\small{{\it School of Computer Science and Engineering,}}\\
\small {{\it  Sun Yat-sen University, Guangzhou 510006, China}}}

\begin{document}

\maketitle
\begin{abstract}
We study the phase discrimination problem, in which we want to decide whether the eigenphase $\theta\in(-\pi,\pi]$ of a given eigenstate $\ket{\psi}$ with eigenvalue $e^{i\theta}$ is zero or not, using applications of the unitary $U$ provided as a black box oracle.
We propose a quantum algorithm named {\it quantum  phase discrimination(QPD)} for this task, with optimal query complexity $\Theta(\frac{1}{\lambda}\log\frac{1}{\delta})$ to the oracle $U$, where $\lambda$ is the gap between zero and non-zero eigenphases and $\delta$ the allowed one-sided error.
The quantum circuit is simple, consisting of only one ancillary qubit and a sequence of controlled-$U$ interleaved with single qubit $Y$ rotations, whose angles are given by a simple analytical formula.
Quantum phase discrimination could become a fundamental subroutine in other quantum algorithms, as we present two applications to quantum search on graphs:
\begin{enumerate}
    \item \textbf{Spatial search on graphs.} Inspired by the structure of QPD, we
propose a new quantum walk model, and based on them we tackle the  spatial search problem, obtaining a novel quantum search algorithm. For any graph with any number of marked vertices, the quantum algorithm that can find a marked vertex with probability $\Omega(1)$ in total evolution time $ O(\frac{1}{\lambda \sqrt{\varepsilon}})$ and query complexity $ O(\frac{1}{\sqrt{\varepsilon}})$, where $\lambda$ is the gap between the zero and non-zero eigenvalues of the graph Laplacian and $\varepsilon$ is a lower bound on the proportion of marked vertices.
    
    \item \textbf{Path-finding on graphs.} By using QPD, we reduce the query complexity of a path-finding algorithm proposed by Li and Zur [arxiv: 2311.07372] from $\tilde{O}(n^{11})$ to $\tilde{O}(n^8)$, in a welded-tree circuit graph with $\Theta(n2^n)$ vertices.
\end{enumerate}
Besides these two applications, we argue that more quantum algorithms might benefit from QPD.
\end{abstract}

\section{Introduction}
Quantum phase estimation~\cite{phase_estimation} is a fundamental subroutine used widely in quantum computation, such as Shor's order-finding algorithm~\cite{Shor} reinterpreted in~\cite{Mosca98}, quantum approximate counting~\cite{amplitude_amplification}, HHL algorithm~\cite{HHL} for quantum linear-systems problem, quantum walk search on Markov chains~\cite{MNRS}, estimating ground-state energies in quantum chemistry~\cite{molecular05}, quantum Metropolis sampling~\cite{Temme2011}, calculating Betti numbers in topological data analysis~\cite{Lloyd2016}, and many others.
The task of phase estimation is to estimate the eigenphase $\theta$ of a given eigenstate $\ket{\psi}$ of a unitary $U$ with unknown eigenvalue $e^{i\theta}$.
In this work, we study a highly related problem named \textit{phase discrimination} stated formally below.

\begin{problem}\label{def:phase}
    Suppose $\ket{\psi}$ is an eigenvector of a unitary $U$ with corresponding eigenvalue $e^{i\phi}$.
    The task is to distinguish the following two cases:
    \begin{equation*}
        \mathrm{(i)}\, \phi=0; \quad \mathrm{or}\quad \mathrm{(ii)}\,  |\phi|\in [\lambda,\pi].
    \end{equation*}
    The eigenvector $\ket{\psi}$ is provided as a quantum state, and the goal is to minimize the number of applications of the unitary $U$ (and possibly $U^\dagger$ and their controlled versions).
\end{problem}

An algorithm is said to solve Problem~\ref{def:phase} with bounded error $\delta$ if it outputs the correct answer with probability greater than $1-\delta$, and it is said to have one-sided error $\delta$ if it additionally does not err in the case of $\phi=0$.

A trivial approach to solve the above problem with bounded error $\delta$ is to perform quantum phase estimation (QPE)~\cite{phase_estimation} with enough precision.
This would, however, require $O(\frac{1}{\delta\lambda})$ calls to controlled-$U$ and $O(\log(\frac{1}{\delta\lambda}))$ ancillary qubits.
The dependence on $\delta$ can be improved by repeated QPE without intermediate measurement~\cite[Theorem 6]{MNRS}, which uses $O(\log(\frac{1}{\delta})\frac{1}{\lambda})$ calls to controlled-$U$ and $O(\log(\frac{1}{\delta})\log(\frac{1}{\lambda}))$ ancillary qubits.
If intermediate measurement is allowed, we only need $O(\log(\frac{1}{\lambda}))$ ancillary qubits that are reused in each repetition, but this may be impractical when it is used as a subroutine in other quantum algorithms.


\subsection{Quantum phase discrimination}
In this work, we present an efficient and concise quantum algorithm for the phase discrimination problem, and then apply the algorithm to settle two search problems on graphs.
\begin{theorem}\label{lem:phase_discrimination}
    For any $\delta\in(0,1)$, there is a quantum algorithm as shown in Fig.~\ref{fig:phase_discrimination} that solves Problem~\ref{def:phase} with one-sided error $\delta$, using $L = O(\frac{1}{\lambda}\ln\frac{1}{\delta})$ controlled-$U$ and one ancillary qubit. The angle parameters $\{\theta_i\}_{i=0}^{L-1}$ are given by an analytical formula.
\end{theorem}

\begin{figure}[hbt]
	\centering
	\includegraphics[width=0.85\textwidth]{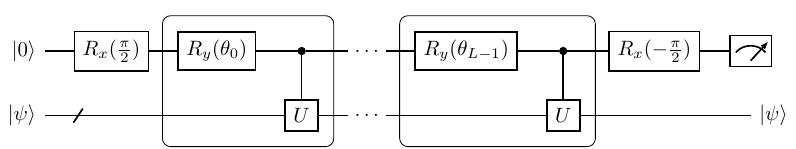}
	\caption{\label{fig:phase_discrimination} Quantum circuit for the phase discrimination problem.}
\end{figure}

We call the quantum algorithm in Theorem~\ref{lem:phase_discrimination} as \textit{quantum phase discrimination} (QPD).
Using QPD, we can not only improve existing quantum algorithms, but also develop new paradigm of quantum search algorithm.
Before delving into these applications of QPD, we first show that its query complexity $O(\frac{1}{\lambda}\ln\frac{1}{\delta})$ achieves optimal scaling.

\begin{theorem}\label{thm:QPD_lower}
    Any quantum algorithm solving Problem~\ref{def:phase} with bounded error $\delta\in(0,1)$ needs $\Omega(\frac{1}{\lambda}\log\frac{1}{\delta})$ applications of $U$ and $U^\dagger$ in total.
\end{theorem}

The proof of Theorem~\ref{thm:QPD_lower} is inspired by \cite[Claim 24]{tight_bounds_phase}, where Mande and de Wolf considered a slightly different problem of distinguishing a specific set of unitaries.

\noindent\textbf{Comparison to other trivial methods.}
To the best of our knowledge, there is no literature that discusses the phase discrimination problem specifically.
There are some off-the-shelf techniques that can be used to solve the phase discrimination problem, but the outcomes are unsatisfactory.
As mentioned earlier, the straightforward approach of repeated QPE without intermediate measurement~\cite[Theorem 6]{MNRS} requires $O(\log(\frac{1}{\delta})\log(\frac{1}{\lambda}))$ ancillary qubits.
Besides QPE, we are aware of two other approaches that can solve the phase discrimination problem (details presented in Appendix~\ref{app:other_methods}): (i) Quantum phase processing (QPP) proposed by Wang et al.~\cite{phase_processing}; and (ii) the eigenstate filtering method~\cite{QLSS_20} used in Lin and Tong's quantum linear system solver.
These two methods can be seen as a variant of QSP, and they both need a numerical and computationally intensive step for finding the angle parameters.
We summarize these trivial methods in Table~\ref{tab:summary}, highlighting the advantages of  QPD:
\begin{enumerate}
    \item It uses only $1$ ancillary qubit.
    \item It does not need the application of controlled-$U^\dagger$.
    As noted in \cite[Section 5.1.2]{property_testing_survey}, the assumption of access to $U^\dagger$ may not be reasonable in an adversarial scenario where we only assume access to $U$ as a black box.
    \item It is easy to implement, as the angle parameters are given by a simple analytical formula.
    The computationally intensive step of finding angle parameters in QSP is often a bottleneck to its scalability and practicality.
    The analytical formula in our QPD significantly reduces classical overhead and streamlines circuit design.
\end{enumerate}

\begin{table}[hbt]
\caption{Comparison to other trivial methods.}\label{tab:summary}
\centering
\begin{tabular}{@{}llll@{}}
\toprule
methods & ancillary qubits  & angle parameters & oracle forms\\
\midrule
repeated QPE~\cite{phase_estimation,MNRS} & $O(\log(\frac{1}{\delta})\log(\frac{1}{\lambda}))$\tablefootnote{It can be decreased to $O(\log(\frac{1}{\lambda}))$ if intermediate measurement is allowed (i.e. reuse the ancillary qubits), but this may be impractical when it acts as a subroutine in other quantum algorithms.}   & None  & controlled-$U$  \\
QPP~\cite{phase_processing} & $1$ & Numerical & controlled-$U$, controlled-$U^\dagger$ \\
eigenstate filtering~\cite{QLSS_20} & $2$ & Numerical & controlled-$U$, controlled-$U^\dagger$ \\
QPD, Theorem~\ref{lem:phase_discrimination}	& $1$ & Analytical	& controlled-$U$ \\
\bottomrule
\end{tabular}
\end{table}

We now present two applications of QPD to quantum search on graphs.

\subsection{Application to spatial search on graphs}
We first apply QPD to spatial search on graphs.
Spatial Search is the problem of finding an unknown marked vertex on a graph $G$.
It is an active research direction in quantum computing and one of the most important algorithmic applications of quantum walks.
Quantum walks are classified into discrete-time quantum walks (DTQW) and continuous-time quantum walks (CTQW).
CTQW evolves a Hamiltonian $H$ for some time $t$, i.e. executes $e^{iHt}$, where $t$ can be any positive real number, whereas DTQW evolves the system according to a unitary that complies with the graph topology for a discrete number of steps.

Over the past 20 years, the research on quantum spatial search algorithms can be roughly summarized into two lines~\cite{universal}.
The first line is to search on specific graphs, aiming at designing a quantum algorithm with time $O(\sqrt{N})$ for a given graph with $N$ vertices. It has been investigated on plenty of different graphs such as $d$-dimensional grids~\cite{coins_05}, hypercube graphs~\cite{CG_04,hypercube_09}, strongly regular graphs~\cite{Janmark2014}, complete bipartite graphs~\cite{bipartite_19,PhysRevA.106.052207}, balanced trees~\cite{RN3}, Johnson graphs~\cite{RN4,johnson_discrete}, Hamming graphs~\cite{HG}, Grassmann graphs~\cite{HG}, Erd\H{o}s-Renyi random graphs~\cite{Chakraborty2016}, and so on.
The quantum spatial search algorithms for those graphs can be designed via DTQW or CTQW.
It is worth pointing out deterministic quantum search algorithms have been proposed for most of the mentioned graphs in \cite{wang2025unifying}  via alternating quantum walks.
Note that discrete and continuous walks have essential differences in the quantum setting, so designing a quantum algorithm in one model does not necessarily lead to a quantum algorithm in the other model.
The relationship between DTQW and CTQW can be found in the literature~\cite{Childs2010}.

The second line of research on quantum spatial search is not limited to specific graphs but is based on Markov chains to answer a more general question: Can quantum walks always provide a quadratic speedup over  classical walks for the spatial search problem?
A breakthrough result on this problem was proved in 2020~\cite{quadratic_20}: For any graph, if there exists a classical search algorithm with time $O(T)$, then we can construct  a quantum search algorithm with time $\widetilde{O}(\sqrt{T})$.
This result was proven based on DTQW.
Subsequently, a similar result based on CTQW was proven in~\cite{quadratic_22}.
These two results have provided a more comprehensive understanding on the application of quantum walks to spatial search problems.
However, it should be noted that the research along the second line  cannot replace the one along the first line.
The reason is that despite these elegant results~\cite{quadratic_20, quadratic_22}, when dealing with specific graphs, the optimal quantum search algorithm remains unknown, and we still need to fully utilize topological properties of the graph to design algorithms.

\textbf{Controlled intermittent quantum walk.}
Inspired by the structure of QPD, we propose a new quantum walk model named controlled intermittent quantum walk (CIQW), which interleaves controlled CTQW $e^{iLt}$ with unitary operations that adjust the control signal, where $L$ is the Laplacian matrix of a simple undirected graph $G$.
An illustration of the CIQW $W$ with single control qubit is shown in Fig.~\ref{fig:pre_QW_model}.

\begin{figure}[hbt]
	\centering
	\includegraphics[width=\textwidth]{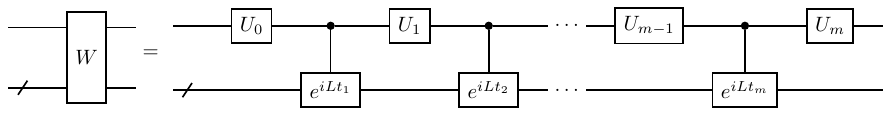}
	\caption{\label{fig:pre_QW_model} Illustration of the CIQW model with single control qubit. }
\end{figure}

This model can be seen as a hybrid of CTQW and DTQW, or more specifically, a continuous version of the MNRS framework~\cite{MNRS} by replacing their controlled DTQW on Markov chain with controlled CTQW $e^{iLt}$.

Based on the CIQW model combined with QPD, we obtain the following result:
\begin{theorem}\label{pre:thm:unknown}
    For any simple undirected graph $G$ with at least $\varepsilon\in(0,1)$ proportion of marked vertices and $\lambda$ being the gap between the zero and non-zero eigenvalues of the graph Laplacian $L$, there is a CIQW-based quantum algorithm that finds a marked vertex with constant success probability, query complexity $O(\frac{1}{\sqrt{\varepsilon}})$ to the oracle that checks whether a vertex is marked or not, and a total of $O(\frac{1}{\lambda \sqrt{\varepsilon}})$ evolution time of the controlled CTQW $e^{iLt}$.
\end{theorem}

Given a graph with Laplacian matrix $L$ and marked set $M$, our CIQW-based algorithm starts from the initial state $\ket{\pi}$, the uniform superposition of all the vertices, and then performs the CIQW and the oracle $e^{i\pi \Pi_M}$ alternately.
The oracle $e^{i\pi \Pi_M}$ multiplies a relative phase of $(-1)$ to vertices $\ket{v}$ in $M$.
An illustration of the algorithm is shown in Fig.~\ref{fig:QWS}.

\begin{figure}[hbt]
	\centering
	\includegraphics[width=0.8\textwidth]{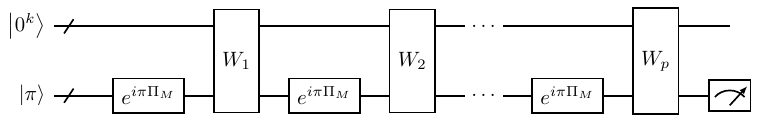}
	\caption{\label{fig:QWS} Illustration of the CIQW-based algorithm. }
\end{figure}

The main idea of our CIQW-based algorithm is to mimic Grover's search algorithms, or more precisely, to implement Grover's iteration $e^{i\pi\ket{\pi}\bra{\pi}} \cdot e^{i\pi \Pi_M}$.
As the oracle $e^{i\pi \Pi_M}$ is already provided, the key is to implement approximately the reflection $e^{i\pi\ket{\pi}\bra{\pi}}$ around the initial state $\ket{\pi}$ using CIQW.
This is where QPD comes into play.
Since $\ket{\pi}$ is the only eigenstate of $L$ with eigenvalue $0$, we can use QPD to distinguish $\ket{\pi}$ from the other eigenstates of $L$, and then construct an approximate reflection around $\ket{\pi}$.

A comparison of our result (Theorem~\ref{pre:thm:unknown}) with some existing quantum algorithms for spatial search is summarized in Table~\ref{tab:summary_spatial_search}.
For a random walk with transition matrix $P$, let $s$ be its stationary distribution, i.e. $s P =s$.
The expected number of steps before the random walk hits a marked vertex, starting from the stationary distribution $s$, is its \textit{hitting time}, denoted by $\mathrm{HT}$.
The gap between $1$ and the other eigenvalues of the transition matrix $P$ is denoted by $\delta\in(0,1)$.
In the electric network model, a random walk takes place on a simple undirected graph with non-negative weighted edges, and its transition probability is proportional to the edge weight.
The quantity $C_\sigma = W\cdot R_\sigma$, where $W$ is the total weight of the graph and $R_\sigma$ is the effective resistance between $\sigma$ and the set of marked vertices~\cite{belovs2013quantum}.
It is known that reversible Markov chains are equivalent to random walks on weighted undirected graphs~\cite{levin2017markov}, and $2C_\sigma =\mathrm{HT}$ when $\sigma=s$~\cite{belovs2013quantum}, and $\frac{1}{\varepsilon} \leq \mathrm{HT} \leq \frac{1}{\delta \varepsilon}$~\cite{unified}.

\begin{table}[hbt]
\caption{Comparison of different quantum spatial search algorithms.}\label{tab:summary_spatial_search}%
\centering
\begin{tabular}{@{}llll@{}}
\toprule
models & initial state  & update cost & checking cost \\
\midrule
AGJK~\cite{quadratic_20} & $\ket{s}$ of $P$   & $\sqrt{\mathrm{HT}}$ steps  & $\sqrt{\mathrm{HT}}$  \\
MNRS~\cite{MNRS} & $\ket{s}$ of $P$ & $\frac{1}{\sqrt{\delta\varepsilon}}$ steps & $\frac{1}{\sqrt{\varepsilon}}$ \\
electric network~\cite{belovs2013quantum,unified} & arbitrary $\ket{\sigma}$ & $\sqrt{C_{\sigma}}$ steps & $\sqrt{C_{\sigma}}$ \\
ACNR~\cite{quadratic_22} & $\ket{s}$ of $P$ & $\sqrt{\mathrm{HT}}$ time & $\sqrt{\mathrm{HT}}$ \\
this work, Theorem~\ref{pre:thm:unknown}	& $\ket{\pi}$ of $L$ & $\frac{1}{\lambda \sqrt{\varepsilon}}$ time	& $\frac{1}{\sqrt{\varepsilon}}$ \\
\bottomrule
\end{tabular}
\end{table}

From Table~\ref{tab:summary_spatial_search}, our algorithm is different from existing quantum spatial search algorithms, and it may add a new perspective to the growing body of quantum search algorithms.
Below are some more remarks on the features of our algorithm:
\begin{enumerate}
    \item It gives a concrete quantum search algorithm for any graph.
    The results of AGJK~\cite{quadratic_20}, MNRS~\cite{MNRS} and ACNR~\cite{quadratic_22} are based on a classical random walk with transition matrix $P$, and the electric network~\cite{belovs2013quantum,unified} assumes an undirected graph with weighted edges.
    Our result gives an explicit quantum algorithm, based only on the Laplacian matrix $L$ of the given graph.

    \item It may achieve better query complexity.
    The results of AGJK~\cite{quadratic_20} and ACNR~\cite{quadratic_22} are state of the art.
    The query complexity to the Check operation by AGJK~\cite{quadratic_20} is $\tilde{O}(\sqrt{{\rm HT}})$, which could be worse than our $1/\sqrt{\varepsilon}$, since $1/\varepsilon \leq {\rm HT}$.
    For example, in the case of $n$-cycle\footnote{ $n$-cycle is an $n$-vertex graph consisting of a single cycle and every vertex has exactly two edges incident with it.} with a single marked vertex, ${1}/{\varepsilon}=n$ while ${\rm HT}=\Theta(n^2)$ (See Appendix~\ref{app:n_cycle} for details). Thus, our algorithm greatly reduces the query complexity in certain cases. 
    The query complexity of the CTQW algorithm of ACNR~\cite{quadratic_22} is not given explicitly, but since the overall expected run-time of their algorithm is $\tilde{O}(\sqrt{\rm HT})$ and the evolving CTQW Hamiltonian $H_{P(s)}$ depends on the interpolated Markov chain $P(s)=(1-s)P+sP'$, where $s\in (0,1)$, in which the absorbing walk $P'$ requires knowledge of the marked vertices, the overall query complexity also scale as $\tilde{O}(\sqrt{\rm HT})$.

    \item The Hamiltonian in our CIQW model is different from that of ACNR~\cite{quadratic_22}, which takes the complicated form of $H = i[U_{P(s)}^\dagger S U_{P(s)}, I\otimes |0\rangle \langle 0|]$, where $U_{P(s)} \ket{x,0} = \sum_{y\in V} \sqrt{P(s)_{xy}} \ket{x,y}$, $\Pi_0 = I \otimes \ket{0}\bra{0}$, $S$ is the swap operation, and $[A,B] = AB-BA$ is the commutator.
    It is remarked in \cite{one_marked_CTQW} (surrounding Eq.~(45) there) that this particular construction of $H$ quadratically amplifies the spectral gap of $P$ and is crucial for subsequent quantum speedup of spatial search.

    \item Our update cost ($\frac{1}{\lambda \sqrt{\varepsilon}}$ time) is incomparable to that ($\frac{1}{\sqrt{\delta\varepsilon}}$ steps) of MNRS~\cite{MNRS}, because the eigenvalue gap $\lambda$ of graph Laplacian $L$ and the eigenvalue gap $\delta$ of transition matrix $P$ are incomparable in most cases.
    In addition, DTQW based on transition matrix $P$ is incomparable to CTQW $e^{iLt}$ with graph Laplacian $L$.
\end{enumerate}

\subsection{Application to path-finding on graphs}
In a recent paper by Li and Zur, a welded tree circuit graph $G$ with $\Theta(n2^n)$ vertices is constructed~\cite[Section 5.3]{multi_electric}.
For the convenience of the reader, the detailed definition of $G$ is shown in Appendix~\ref{app:circuit_graph}.
Based on this graph, they consider the following path-finding problem.
\begin{problem}[Problem 5.1 in \cite{multi_electric}]\label{def:path-finding}
    Given an adjacency list oracle $O_G$ to the welded tree circuit graph $G$, the goal is to find an $s$-$t$ path in $G$, where $s$ is the only vertex in $G$ with degree $2$ and $t$ is the only vertex in $G$ with degree $1$.
\end{problem}

It is shown in \cite[Theorem 5.6]{multi_electric} that any classical algorithm making at most $2^{n/6}$ queries to $O_G$ solves Problem~\ref{def:path-finding} with probability $2^{-O(n)}$, whereas there exists a quantum algorithm that solves Problem~\ref{def:path-finding} with success probability $1- O(\delta)$ and $O(n^{11}\log(n/\delta))$ queries~\cite[Theorem 5.2]{multi_electric}.
Thus, an exponential speedup of the path-finding problem on this graph is provided.

Their quantum algorithm relies on a subroutine (see Lemma~\ref{lem:phase_filter} in Section~\ref{sec:path}) that performs QPE on the quantum walk operator $U$ to approximate the projection to the desired $1$-eigenspace of $U$.
We replace QPE with our QPD, leading to a reduction in query complexity by a factor of $\widetilde{O}(n^3)$.
Overall, we obtain the following result.
\begin{theorem}
    There exists a quantum algorithm that solves Problem~\ref{def:path-finding} with success probability $1- O(\delta)$ and $O(n^{8}\log(n)\log(n/\delta))$ queries.
\end{theorem}

\textbf{Other potential applications of QPD.} QPD can also be used for filtering out undesired eigenstate~\cite{zihao}, with the same benefit that parameters in the quantum circuit are given analytically compared to existing method~\cite{QLSS_20}.
Technically, the effect of QPD followed by a projection $\ket{0}\bra{0}$ on the ancillary qubit is approximately a projection to the $e^{i0}$-eigenspace of the unitary $U$.
This leads to its applications in various filtering or projection-based quantum algorithms, such as quantum linear system solver~\cite{QLSS_20,QLSS_22}, ground state preparation~\cite{Lin2020nearoptimalground,ground_state_22}, estimating normalized Betti numbers~\cite{QTDA_PRX_Quantum_24} in quantum topological data analysis~\cite{Lloyd2016}, and perhaps many others.

\subsection{Paper organization}
The rest of this paper is organized as follows.
In Section~\ref{sec:phase_discrimination}, we formally present QPD (Theorem~\ref{res:lem:phase_discrimination}) and its proof.
The proof of the lower bound on phase estimation (Theorem~\ref{thm:QPD_lower}) is presented in Section~\ref{subsec:QPD_lower}.
We then discuss in detail two applications of QPD, one to spatial search on graphs (Section~\ref{sec:CIQW}), and the other to path-finding on graphs (Section~\ref{sec:path}).
Finally, we conclude and point out some future work in Section~\ref{sec:discussion}.

\section{Quantum phase discrimination}\label{sec:phase_discrimination}

Our QPD is formally presented in the theorem below, and it features analytical angle parameters which have also appeared in fixed-point quantum search~\cite{fixed_point,quasi_chebyshev}.
The relation between these two seemingly different tasks from the perspective of function approximation is discussed in Remark~\ref{rem:relation}.

\begin{theorem}[Theorem~\ref{lem:phase_discrimination}, formal]\label{res:lem:phase_discrimination}
    Assume $\ket{\psi}$ is an eigenvector of the unitary $U$ such that $U\ket{\psi} =e^{i\phi}\ket{\psi}$ and $\phi\in(-\pi,\pi]$.
    Consider the quantum circuit $C(U,\lambda ,L)$ shown in Fig.~\ref{res:fig:phase_discrimination}, where $\lambda \in (0,\pi)$, $L$ is an odd integer, and $\theta_n = 2\arctan[\sin(\frac{\lambda }{2})\tan(\frac{n}{L}\pi)]$ for $n \in\{0,\dots,L-1\}$.    
    Then the final state $\ket{w}\ket{\psi}$ satisfies
    \begin{equation}\label{eq:0_omega}
        \left| \braket{0|w} \right| = \left| \frac{T_L\left({\cos(\frac{\phi}{2})}/{\cos(\frac{\lambda }{2})}\right)} {T_L({1}/{\cos(\frac{\lambda }{2})})} \right|,
    \end{equation}
    where $T_L(x)$ is the Chebyshev polynomial of the first kind~\footnote{The polynomial $T_L(x)$ satisfies the recurrence relation $T_0(x)=1, T_1(x)=x, T_{n+1}(x) = 2x T_n(x) -T_{n-1}(x)$, and has an explicit formula $T_L(x) = \cos(L\arccos(x))$ for $|x|\leq 1$; $T_L(x) = \cosh(L\,\mathrm{arccosh}(x))$ for $x \geq 1$; $T_L(x) = (-1)^L \cosh(L\,\mathrm{arccosh}(-x))$ for $x \leq -1$.}.
    
    Furthermore, if $\delta\in(0,1)$ is the allowable one-sided error, then by setting $L\geq \frac{\ln(2/\delta)}{\lambda/2}$, we have:
    \begin{subnumcases}{}
        \ket{w}=\ket{0}, \quad \ \, {\rm if}\, \phi=0; \label{eq:0_omega_1}\\
        \left| \braket{0|w} \right| \leq \delta, \quad {\rm if}\, \left|\phi\right|\geq \lambda. \label{eq:0_omega_delta} 
    \end{subnumcases}
\end{theorem}

\begin{figure}[hbt]
	\centering
	\includegraphics[width=0.85\textwidth]{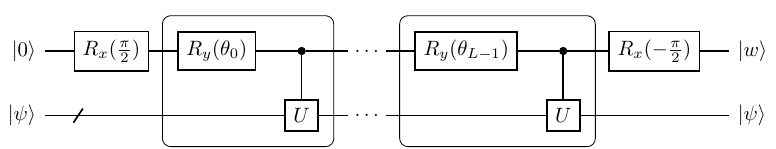}
	\caption{\label{res:fig:phase_discrimination} Quantum circuit $C(U,\lambda ,L)$ of QPD.}
\end{figure}

The function $\left| \braket{0|w} \right|$ of $\phi$ is an approximation of the ``delta'' function that takes value $1$ at $\phi=0$ and $0$ for $\phi\neq 0$.
For an intuitive illustration, we take $\lambda =\frac{\pi}{8}, \delta=0.1$ and $L=17$ such that $L\geq \frac{\ln(2/\delta)}{\lambda/2}$ as an example, and the function graph of $\left| \braket{0|w} \right|$ for $\phi\in(-\pi,\pi]$ is shown in Fig.~\ref{fig:zero_omega}.

\begin{figure}[hbt]
	\centering
	\includegraphics[width=0.65\textwidth]{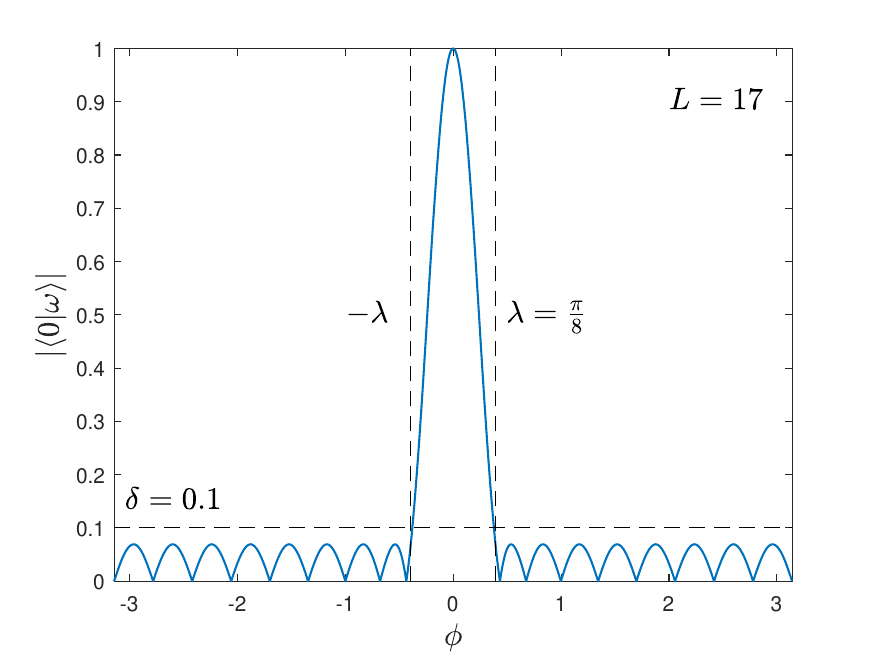}
	\caption{\label{fig:zero_omega} Function graph of $\left| \braket{0|w} \right|$ for $\phi\in(-\pi,\pi]$.
    The expression is shown by Eq.~\eqref{eq:0_omega},
    where $\lambda =\frac{\pi}{8}, \delta=0.1$ and $L=17$ such that $L\geq \frac{\ln(2/\delta)}{\lambda/2}$.
    The dashed lines show that $\left| \braket{0|w} \right| \leq 0.1$ when $|\phi|\geq \frac{\pi}{8}$.
    Note that $\left| \braket{0|w} \right|=1$ when $\phi=0$.}
\end{figure}

One may notice that the height of the small peaks in Fig.~\ref{fig:zero_omega}, i.e. the actual maximum value of $\left|\braket{0|w}\right|$ when $|\phi|\geq \lambda $, is smaller than the allowable error $\delta$, and is $1/T_L(1/\cos(\lambda /2))$ by Eq.~\eqref{eq:0_omega} and Eq.~\eqref{eq:numer_1} below.
This is because the solution $y$ to the desired equality $1/T_y(1/\cos(\lambda /2)) = \delta$ is usually not an integer, but we want $y$ to be an odd integer $L$.
The issue can be solved as follows, at the expense of a more complicated expression of $\theta_n$ used in the quantum circuit $C(U,\lambda ,L)$.
After setting the odd number $L$ to satisfy $1/T_L(1/\cos(\lambda /2)) \leq \delta$, we can reduce $\lambda $ to $\lambda '$ such that $1/T_L(1/\cos(\lambda '/2)) = \delta$, by the fact that $T_L(x)$ is increasing for $x\in[1,\infty)$.
Using the definition $T_y(x) = \cosh(y\,\mathrm{arccosh}(x))$ and the identity $\cosh^2(x) -1 = \sinh^2(x)$,
we have:
\begin{align}
& 1/T_L(1/\cos(\lambda '/2)) = \delta \\
\Leftrightarrow & 1/\delta = \cosh(L\,\mathrm{arccosh}(1/\cos(\frac{\lambda '}{2}))) \\
\Leftrightarrow & \cosh(\frac{1}{L}\,\mathrm{arccosh}(1/\delta)) = 1/\cos(\frac{\lambda '}{2}) \\
\Leftrightarrow & \sin(\frac{\lambda '}{2}) = \sqrt{1-\frac{1}{\cosh^2(\frac{1}{L}\,\mathrm{arccosh}(1/\delta))}}\\
\Leftrightarrow & \sin(\frac{\lambda '}{2}) =\tanh(\frac{1}{L}\,\mathrm{arccosh}(1/\delta)).
\end{align}
As mentioned earlier, this makes the expression of $\theta_n = 2\arctan[\sin(\frac{\lambda '}{2})\tan(\frac{n}{L}\pi)]$ more complicated.

\subsection{Proof of Theorem~\ref{res:lem:phase_discrimination}}
At the heart of Theorem~\ref{res:lem:phase_discrimination} as well as the fixed-point quantum search~\cite{fixed_point} lies the following quasi-Chebyshev lemma, whose full proof has only recently been shown in~\cite{quasi_chebyshev}.

\begin{lemma}\label{lem:chebyshev}
    Suppose $\gamma \in (0,1]$. For any given odd number $L$, consider the odd polynomial $a^\gamma_L(x)$ of degree $L$ defined by the following recurrence relation:
\begin{align}
a^\gamma_0(x) &= 1,\ a^\gamma_1(x)=x, \\
a^\gamma_{n+1}(x) &= x(1+e^{-i\theta_n}) a^\gamma_n(x) -e^{-i\theta_n} a^\gamma_{n-1}(x), \label{eq:a_L_recur}
\end{align}
    where the angles are 
\begin{equation}\label{eq:phi_n_def}
\theta_n = 2\arctan\left( \sqrt{1-\gamma^2} \tan(\frac{n}{L}\pi) \right), \ n\in\{1,\dots,L-1\}.
\end{equation}
    Then the explicit formula of $a^\gamma_L(x)$ is
\begin{equation}\label{eq:a_gamma_L}
    a^\gamma_L(x) = T_L(x/\gamma) / T_L(1/\gamma).
\end{equation}
\end{lemma}

\begin{remark}[connection between QPD and fixed-point quantum search]\label{rem:relation}
    From the perspective of function approximation, the two seemingly different tasks of phase discrimination and fixed-point quantum search are both deeply connected to the odd polynomial $a_L^\gamma(x)$ shown in Eq.~\eqref{eq:a_gamma_L}.
    In the former task, we want to approximate the ``delta" function.
    This can be achieved by letting $x=\cos(\phi/2)$ and $\gamma=\cos(\lambda/2)$ in $a_L^\gamma(x)$ and setting $L\geq 2\ln(2/\delta)/\lambda$, and we will soon prove that the desired function $f(\phi):= a_L^\gamma(x)$ satisfies $f(0)=1$ and $|f(\phi)| \leq \delta$ when $|\phi|\leq \lambda$.
    In the latter task, we want to approximate the ``sign" function that has value $1$ (i.e. the highest success probability achievable) whenever $\lambda>0$ (i.e. the marked elements are non-empty).
    Suppose a lower bound $w^2$ on the proportion $\lambda^2$ of marked elements is known in advance. 
    Let $x=\sqrt{1-\lambda^2}$, and $\gamma=\sqrt{1-w^2}$ in $a_L^\gamma(x)$, and set $L \geq \ln(2/\delta)/w$, then it can be shown similarly~\cite{quasi_chebyshev} that the desired function $P(\lambda) := \sqrt{1-a_L^\gamma(x)^2}$ satisfies $P(0)=0$ and $P(\lambda)\geq \sqrt{1-\delta^2}$ as long as $\lambda > w$.
\end{remark}

We now prove the correctness of Theorem~\ref{res:lem:phase_discrimination} using Lemma~\ref{lem:chebyshev}.
Recall that the three single-qubit operations that correspond to the rotations around the $x,y,z$ axes respectively by angle $\theta$ on the Bloch sphere are as follows:
\begin{equation}\label{eq:rotation_axis_def}
R_x(\theta) = 
\begin{bmatrix}
\cos(\frac{\theta}{2}) & -i\sin(\frac{\theta}{2}) \\
-i\sin(\frac{\theta}{2}) & \cos(\frac{\theta}{2})
\end{bmatrix},
R_y(\theta) =
\begin{bmatrix}
\cos(\frac{\theta}{2}) & -\sin(\frac{\theta}{2}) \\
\sin(\frac{\theta}{2}) & \cos(\frac{\theta}{2})
\end{bmatrix},
R_z(\theta) =
\begin{bmatrix}
e^{-i\theta/2} & 0 \\
0 & e^{i\theta/2}
\end{bmatrix}.
\end{equation}
Since $U\ket{\psi} = e^{i\phi}\ket{\psi}$, after applying controlled-$U$ to $\ket{b}\ket{\psi}$ for $b\in\{0,1\}$, we obtain $e^{ib\phi}\ket{b}\otimes\ket{\psi}$.
By the definition of $R_z(\theta)$ shown in Eq.~\eqref{eq:rotation_axis_def}, we know it is equivalent to applying $e^{i\frac{\phi}{2}} R_z(\phi)$ to the first qubit and leaves the state $\ket{\psi}$ in the second register unchanged.
Therefore, after applying the quantum circuit $C(U,\lambda ,L)$ shown in Fig.~\ref{res:fig:phase_discrimination} to the initial state $\ket{0}\ket{\psi}$, we obtain the final state $\ket{w}\ket{\psi}$, where
\begin{equation}\label{eq:omega_1}
\ket{w} = e^{iL\frac{\phi}{2}}\cdot R_x(-\frac{\pi}{2}) \cdot  \prod_{n=0}^{L-1} \big( R_z(\phi) \cdot R_y(\theta_n) \big) \cdot R_x(\frac{\pi}{2}) \ket{0}.
\end{equation}
It can be easily verified that the following two identities hold:
\begin{align}
R_x(-\frac{\pi}{2}) \cdot R_z(\phi) \cdot R_x(\frac{\pi}{2}) &= R_y(\phi), \label{eq:xzx_y}\\
R_x(-\frac{\pi}{2}) \cdot R_y(\theta) \cdot R_x(\frac{\pi}{2}) &= R_z(-\theta). \label{eq:xyx_z}
\end{align}
Thus by inserting $R_x(\frac{\pi}{2}) \cdot R_x(-\frac{\pi}{2}) =I$ to Eq.~\eqref{eq:omega_1}, and using Eqs.~\eqref{eq:xzx_y}, \eqref{eq:xyx_z}, we have:
\begin{equation}\label{eq:omega_2}
    \ket{w} = e^{iL\frac{\phi}{2}} \cdot \prod_{n=0}^{L-1} \big( R_y(\phi) \cdot R_z(-\theta_n) \big) \ket{0}.
\end{equation}
Our motivation for switching the rotation axes of $R_z(\phi)$ and $R_y(\theta_n)$ is that the new $R_y(\phi)$ and $R_z(\theta_n)$ are closer to the formalism in fixed-point quantum search~\cite{fixed_point,quasi_chebyshev}, where $R_y(\phi)$ corresponds to the unitary $\mathcal{A}$ that prepares a linear combination of marked states and unmarked states, and $R_z(-\theta_n)$ corresponds to the phase oracle that multiplies a relative phase shift to the marked states.

We first consider the case where $\phi=0$.
Thus $R_y(\phi)=I$, and Eq.~\eqref{eq:omega_2} becomes
\begin{equation}
    \ket{w} = R_z\left(-\sum_{n=0}^{L-1}\theta_n\right) \ket{0} = \exp(i\sum_{n=0}^{L-1}\frac{\theta_n}{2}) \ket{0} =\ket{0},
\end{equation}
where we have used the fact that $\theta_{L-n}=-\theta_n$ for $n\in \{1,\dots,L-1\}$ by the definition of $\theta_n$.
Therefore, we obtain Eq.~\eqref{eq:0_omega_1}.

We then consider the case where $\phi\neq 0$.
Denote by $A(\theta_n) := R_y(\phi)\cdot e^{-i\frac{\theta_n}{2}} R_z(-\theta_n)$ the product of the two rotations in Eq.~\eqref{eq:omega_2}.
Using their matrix expression shown in Eq.~\eqref{eq:rotation_axis_def}, we know:
\begin{equation}\label{eq:A_theta_mat}
    A(\theta) =
    \begin{bmatrix}
        \cos(\frac{\phi}{2}) & -e^{-i\theta} \sin(\frac{\phi}{2}) \\
        \sin(\frac{\phi}{2}) & e^{-i\theta} \cos(\frac{\phi}{2})
    \end{bmatrix}.
\end{equation}
Let $a_0:=1, b_0:=0$, and $[a_n, \sin(\frac{\phi}{2}) b_n]^T := A(\theta_{n-1}) \cdots A(\theta_0) \ket{0}$.
Then we have
\begin{equation}\label{eq:0_omega_a_L}
    \left| \braket{0|w} \right| = |a_L|.
\end{equation}
From Eq.~\eqref{eq:A_theta_mat}, we obtain the following recurrence relation of $a_n, b_n$:
\begin{align}
a_{n+1} &= a_n\cos(\frac{\phi}{2})  -b_n e^{-i\theta_n} \sin^2(\frac{\phi}{2}), \label{eq:an_bn_1} \\
b_{n+1} &= a_n +b_n e^{-i\theta_n} \cos(\frac{\phi}{2}), \label{eq:an_bn_2}
\end{align}
where $n\in\{0,\dots,L-1\}$.
To decouple the above recurrence relation to obtain the recurrence relation of $a_n$, we first use the linear combination of the above equations, i.e. $\eqref{eq:an_bn_1} \cdot \frac{\cos(\phi/2)}{\sin^2(\phi/2)} + \eqref{eq:an_bn_2}$, to eliminate $b_n$ and obtain an expression of $b_{n+1}$ regarding $a_n$ and $a_{n+1}$,
and then let $n \mapsto (n-1)$, obtaining
\begin{equation}\label{eq:bn}
    b_n = \frac{a_{n-1} - a_n \cos(\frac{\phi}{2})}{\sin^2(\frac{\phi}{2})},
\end{equation}
where $n\in\{1,\dots,L\}$.
Substituting Eq.~\eqref{eq:bn} into Eq.~\eqref{eq:an_bn_1}, and letting $x:=\cos(\frac{\phi}{2})$, we obtain the recurrence relation of $a_n$:
\begin{equation}
a_{n+1} = x (1+e^{-i\theta_n}) a_n -e^{-i\theta_n} a_{n-1}, \ n\in\{1,\dots,L-1\},
\end{equation}
where the first two terms are $a_0=1$ and $a_1=x$.
Note that $a_1 \equiv x$ regardless of the value of $\theta_0$, and we let $\theta_0 = 2\arctan[\sin(\frac{\lambda }{2})\tan(\frac{0}{L}\pi)] =0$ for a succinct expression.

We are now only one step away from Eq.~\eqref{eq:0_omega}.
Let
\begin{equation}
    \gamma = \cos(\frac{\lambda }{2}) \in (0,1).
\end{equation}
Then $\sqrt{1-\gamma^2} = \sin(\frac{\lambda }{2})$, and our choice of $\theta_n = 2\arctan[\sin(\frac{\lambda }{2})\tan(\frac{n}{L}\pi)]$ for $n\in\{1,\dots,L-1\}$ coincides with Eq.~\eqref{eq:phi_n_def} in Lemma~\ref{lem:chebyshev}.
Thus by Eq.~\eqref{eq:a_gamma_L} and Eq.~\eqref{eq:0_omega_a_L}, we have:
\begin{equation}\label{eq:0_omega_2}
    \left| \braket{0|w} \right| = \left| \frac{T_L\left({\cos(\frac{\phi}{2})}/{\cos(\frac{\lambda }{2})}\right)} {T_L({1}/{\cos(\frac{\lambda }{2})})} \right|.
\end{equation}
Since the RHS of Eq.~\eqref{eq:0_omega_2} equals $1$ when $\phi=0$, Eq.~\eqref{eq:0_omega_2} also holds for the case where $\phi=0$ (we have proved $\left| \braket{0|w} \right|=1$ in this case).
Thus Eq.~\eqref{eq:0_omega} is proved.

Finally, we prove Eq.~\eqref{eq:0_omega_delta}, i.e. $\left| \braket{0|w} \right| \leq \delta$, when $\left|\phi\right|\geq \lambda $ and $L\geq \frac{\ln(2/\delta)}{\lambda/2}$.
By the fact that $|T_L(x)| \leq 1$ for $|x|\leq 1$, we know
\begin{equation}\label{eq:numer_1}
    \left| T_L\left({\cos(\frac{\phi}{2})}/{\cos(\frac{\lambda }{2})}\right) \right| \leq 1, \ \mathrm{for}\, |\phi|\in [\lambda , \pi].
\end{equation}
Combining Eq.~\eqref{eq:numer_1} with Eq.~\eqref{eq:0_omega_2}, it suffices to show that $T_L({1}/{\cos(\frac{\lambda }{2})}) \geq 1/\delta$.
By the fact that $T_L(x) = \cosh(L\,\mathrm{arccosh}(x))$ for $x\geq 1$, it is equivalent to $\cosh(L\,\mathrm{arccosh}(1/\cos(\lambda /2))) \geq 1/\delta$.
Since $\cosh(x)$ is increasing for $x \geq 0$, it is further equivalent to
\begin{equation}\label{eq:L_lower}
L \geq \frac{\mathrm{arccosh}(1/\delta)}{\mathrm{arccosh}(1/\cos(\lambda /2))}.
\end{equation}
Thus to prove Eq.~\eqref{eq:0_omega_delta}, it suffices to show that the RHS of Eq.~\eqref{eq:L_lower} has the following upper bound:
\begin{equation}\label{eq:L_lower_relax}
    \frac{\mathrm{arccosh}(1/\delta)}{\mathrm{arccosh}(1/\cos(\lambda /2))} \leq \frac{\ln(2/\delta)}{\lambda /2}.
\end{equation}
The numerator is upper bounded by $\mathrm{arccosh}(1/\delta) = \ln(1/\delta+\sqrt{1/\delta^2-1}) \leq \ln(2/\delta)$, where we use the definition  $\mathrm{arccosh}(x) = \ln(x+\sqrt{x^2-1})$ for $x \geq 1$ in the first equality.
The denominator is lower bounded by $\mathrm{arccosh}(1/\cos(\lambda /2)) = \ln(\frac{1+\sin(\lambda /2)}{\cos(\lambda /2)}) \geq \lambda /2$, where the last inequality uses the fact that $f(x):= \ln(\frac{1+\sin(x)}{\cos(x)}) - x$ satisfies $f(x)\geq 0$ for $x\in[0,\pi/2)$, which follows from $f(0)=0$ and $f'(x) = 1/\cos(x) -1 \geq 0$ for $x\in[0,\pi/2)$.
This proves Eq.~\eqref{eq:L_lower_relax}, and thus we have finished the proof of Theorem~\ref{res:lem:phase_discrimination}.

\subsection{proof of Theorem~\ref{thm:QPD_lower}}\label{subsec:QPD_lower}

For a value $\theta\in(-\pi,\pi]$, Mande and de Wolf~\cite{tight_bounds_phase} define the unitary $U_\theta$ (with dimension greater than $1$) as
\begin{equation}
	U_\theta := I-(1-e^{i\theta})\ket{0}\bra{0}.
\end{equation}
Note that $\ket{0}$ is the eigenvector of $U_\theta$ with eigenphase $\theta$.
The task $\mathrm{dist}_{\lambda,\delta}$ (cf. Definition~5 in \cite{tight_bounds_phase}, with symbols renamed) is to distinguish between $\theta=0$ and $|\theta|\in [\lambda,\pi]$ with probability at least $1-\delta$.
Using a variant of the polynomial method with trigonometric polynomials, it is proved in \cite[Claim 24]{tight_bounds_phase} that for $\delta, \lambda \in (0,1/2)$, every quantum algorithm for $\mathrm{dist}_{\lambda,\delta}$ needs
\begin{equation}\label{eq:lower_bound}
    \Omega\left(\frac{1}{\lambda}\log\frac{1}{\delta}\right)
\end{equation}
applications of controlled-$U_\theta$ and controlled-$U_\theta^\dagger$ in total.
As the eigenvector $\ket{0}$ of $U_\theta$ is \textit{fixed}, the task $\mathrm{dist}_{\lambda,\delta}$ is a special case of the phase discrimination problem, and thus Eq.~\eqref{eq:lower_bound} is also a lower bound on the phase discrimination problem.

\section{Application to spatial search on graphs}\label{sec:CIQW}
Spatial Search is the problem of finding an unknown marked vertex on a graph $G$. When designing algorithms for this problem, it is generally assumed that there is an oracle checking whether a given vertex is the marked one, and the algorithm should invoke this oracle as few times as possible. In quantum computing, the standard oracle works as follows: 
\begin{equation}\label{eq:standard_oracle_def}
    O_M\ket{b}\ket{v}=\ket{b \oplus f(v)}\ket{v},
\end{equation}
for $b \in \{0, 1\}$ and $v \in V(G)$, where the Boolean function $f(v)=1$ iff $v$ is in the marked set $M$.
We will also use the phase oracle $e^{i\pi \Pi_M}$, which works as: 
\begin{equation}\label{eq:general_oracle_def}
    e^{i\pi \Pi_M} \ket{v} =(-1)^{f(v)}\ket{v}.
\end{equation}
The phase oracle $e^{i\pi \Pi_M}$ can be implemented by the standard oracle $O_M$ due to the well-known phase kick-back effect: $O_M\ket{v}\ket{-}=e^{i\pi \Pi_M}\ket{v}\ket{-}$, where $\ket{-}=(\ket{0}-\ket{1})/\sqrt{2}$.

We will consider spatial search on a simple undirected connected graph $G$, where ``simple” means the graph has no loops or multiple edges between any two vertices.
Denote by $V$ the vertex set and $E$ the edge set of $G$.
The Laplacian matrix of $G$ is $L = D-A$, where $D$ is the diagonal matrix with $D_{jj} = \mathrm{deg}(j)$, the degree of vertex $j$, and $A$ is the $(0,1)$ adjacency matrix of $G$, where $A_{ij}=1$ if and only if $(i,j)\in E$.
The Laplacian $L$ of a simple undirected graph $G$ is a symmetric matrix since $D$ and $A$ are both symmetric, and thus $L$ is diagonalizable.

The spectrum of $L$ has the following nice properties:
(i) All the eigenvalues of $L$ are non-negative and bounded above by $N := |V|$~\cite{Laplacian_distinct}.
(ii) $0$ is a simple eigenvalue (i.e. with multiplicity one) of $L$ if $G$ is connected, and the corresponding eigenvector is the uniform superposition of all vertices $\ket{\pi}:= \frac{1}{\sqrt{N}}\sum_{v\in V}\ket{v}$~\cite{Laplacian_spectral}.

\subsection{Controlled intermittent quantum walks}
Suppose a graph $G$ has the Laplacian matrix $L$.
The CIQW on $G$ has the state space $\mathcal{K}\otimes\mathcal{H}$, where $\mathcal{K}$ consisting of $k$ qubits (we only need $k=1$ in our spatial search algorithm) denotes the ancillary space that stores the control signal, and  $\mathcal{H} = \mathrm{span} \{ \ket{v} : v\in V(G) \}$  spanned by all vertices of the graph $G$ is the walking space.
A CIQW (with $m$ intermittent steps), denoted by $W$, is defined as follows:
\begin{equation}\label{eq:CIQW_evolution}
    W := \prod_{j=1}^{m} \left[ (U_{j} \otimes I)\cdot \Lambda_k(e^{iLt_j}) \right]
    \cdot (U_0 \otimes I),
\end{equation}
where $U_j$ for $j\in\{0,1,\dots,m\}$ are unitary operators changing the control signal in ancillary space $\mathcal{K}$,
and 
$\Lambda_k(e^{iLt_j})=\sum_{l} \ket{l}\bra{l}\otimes  e^{ilLt_j}$ denotes the controlled unitary transformation that applies $(e^{iLt_j})^l$ to $\mathcal{H}$, controlled by  $l\in\{0,1,\dots,2^k-1\}$  in the ancillary space $\mathcal{K}$.
We denote $\prod_{j=1}^{m} A_j := A_m A_{m-1} \cdots A_1$ (rather than $A_1 A_2 \cdots A_m$), since $\mathcal{H}$ consists of column vectors and the rightmost $A_1$ should be applied first.
An illustration of the CIQW $W$ is shown in Fig.~\ref{fig:QW_model}.

\begin{figure}[hbt]
	\centering
	\includegraphics[width=\textwidth]{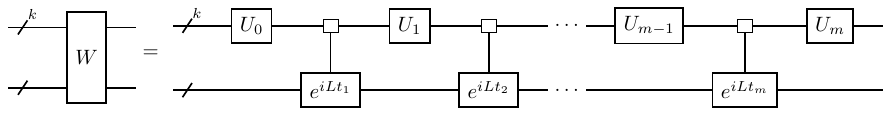}
	\caption{\label{fig:QW_model} Illustration of the CIQW model. }
\end{figure}
As shown in Fig.~\ref{fig:QW_model}, a CIQW performs a controlled CTQW $\Lambda_k(e^{iLt_1})$ for a period of time under a certain control signal generated by $U_0$, then adjusts the control signal by $U_1$ and performs a controlled CTQW $\Lambda_k(e^{iLt_2})$ for another period of time under the new control signal, repeating the operation in this way.
Aside from the query complexity to the oracles, we are also concerned with the total evolution time, which is the sum of all the time $\{t_j\}$ appearing in the CTQW $e^{iLt_j}$.


\subsection{Spatial search via CIQW}
We propose a CIQW-based algorithm for spatial search.
For any general graph with any set of marked vertices, our algorithm can find a marked vertex with bounded-error:
\begin{theorem}[Theorem~\ref{pre:thm:unknown} restatement]\label{thm:unknown}
    For any graph with any number of marked vertices,  a CIQW-based quantum algorithm can find a marked vertex with probability $\Omega(1)$ in total evolution time $ O(\frac{1}{\lambda \sqrt{\varepsilon}})$ and  query complexity $ O(\frac{1}{\sqrt{\varepsilon}})$, where $\lambda$ is the gap between the zero and non-zero eigenvalues of the graph Laplacian and $\varepsilon$ is a lower bound on the proportion of marked vertices.
\end{theorem}

The inverse linear dependence on the eigenvalue gap $\lambda$ of the graph Laplacian $L$ is perhaps optimal, as the inverse process of a CIQW based quantum search algorithm can be seen as a special case of preparing the ground state of a Hamiltonian $H$ (note that $\ket{\pi}$ is the ground state of $L$ from which the search algorithm starts with), the task of which is proven by Tong et al.~\cite{ground_state_22, Lin2020nearoptimalground} to have a lower bound of $\Omega(1/\Delta)$ on the number of applications of $e^{iH}$, where $\Delta$ is the gap between the ground-state energy of $H$ with the rest of its spectrum.

\textbf{Proof overview.}
Given a graph with Laplacian matrix $L$ and marked set $M$, our CIQW-based algorithm starts from the initial state $\ket{\pi}$, the uniform superposition of all the vertices (also the eigenstate of $L$ with eigenvalue $0$), and then performs the CIQW in Fig.~\ref{fig:QW_model} and the oracle $e^{i\pi \Pi_M}$ alternately.
The main idea is to mimic Grover's search algorithms, or more precisely, to implement Grover's iteration $e^{i\pi\ket{\pi}\bra{\pi}} \cdot e^{i\pi \Pi_M}$.
As the oracle $e^{i\pi \Pi_M}$ is already provided by Eq.~\eqref{eq:general_oracle_def}, the key is to implement approximately the reflection $e^{i\pi\ket{\pi}\bra{\pi}}$ around the initial state $\ket{\pi}$ using CIQW.

For general graphs with $N$ vertices, its Laplacian eigenvalues $0 =\lambda_1 < \lambda_2 \leq \cdots \leq \lambda_{N} \leq N$ can be irrational numbers.
For example, the Laplacian eigenvalues of the $n$-vertex cycle are $4\sin(\frac{k\pi}{n})^2$ for $k=0,1,\dots,(n-1)$, since its Laplacian matrix $L=D-A=2I-A$ and the eigenvalues of $A$ are $2\cos(\frac{2\pi k}{n})$ as shown in Appendix~\ref{app:n_cycle}.

The most straightforward approach to implement an approximation of the reflection $e^{i\pi\ket{\pi}\bra{\pi}}$ is to use phase estimation on $e^{iLt}$, but that would require $O(\log(\frac{1}{\delta}) \log(\frac{\lambda_{N}}{\lambda_2}))$ ancillary qubits~\cite[Theorem 6]{MNRS}, where $\delta$ is the error of approximation.
Instead, we will use CIQW shown in Fig.~\ref{fig:QW_model} with $k=1$ to implement the approximate reflection, where a key subroutine is QPD shown in Theorem~\ref{res:lem:phase_discrimination}.
The analysis of the evolution time $t_1 = O(\frac{1}{\lambda_2} \log(\frac{1}{\delta}) )$ to approximate the reflection $e^{i\pi\ket{\pi}\bra{\pi}}$ with error $\delta$ is shown in Section~\ref{subsec:reflection}.
The reduction of ancillary qubit numbers from $O(\log(\frac{1}{\delta}) \log(\frac{\lambda_{N}}{\lambda_2}))$ to only $1$ means that the only signal processing unitaries we need now are single qubit rotations, which makes the implementation much simpler compared to repeated QPE that requires quantum Fourier transform.
As an analytical formula gives the rotation angles, we are also free from the complicated numerical calculations required in other methods shown in Table~\ref{tab:summary}.

After the above treatment, Grover's iteration is implemented with approximation error $\delta$, and a direct approach to finding a marked vertex with constant probability is to iterate the approximate Grover's iteration for $O(1/\sqrt{\varepsilon})$ times, where $\varepsilon$ is the proportion of marked vertices, and set the approximation error to $\delta = O(\sqrt{\varepsilon})$.
However, this would incur a multiplicative logarithmic factor of $O(\log(1/\varepsilon))$ in the total cost, since $t_1 = O(\frac{1}{\lambda_2} \log(\frac{1}{\delta}) )$.
To address this issue, we use recursive amplitude amplification with approximate reflection introduced in the MNRS framework~\cite{MNRS}.
We first consider the case where the proportion of marked vertices is given (Lemma~\ref{lem:recur_known} in Section~\ref{subsec:general}), and then deal with the general case where only a lower bound on the proportion of marked vertices is known (Lemma~\ref{lem:recur_unknown} in Section~\ref{subsec:general}).
We correct a small mistake in the original error analysis (see Remark~\ref{remark:why_pi_4} in Appendix~\ref{app:1} for details), and give the revised full proof in Appendices~\ref{app:1} and \ref{app:2}.

\subsection{Approximate reflection}\label{subsec:reflection}
We now show how to implement an approximation $R(\delta)$ of the reflection $\mathrm{ref}(\pi) = 2\ket{\pi}\bra{\pi}-I$ with one-sided error $\delta$ (stated formally in Lemma~\ref{lem:approx_ref} below), based on our QPD (Theorem~\ref{res:lem:phase_discrimination}). 
The meaning of one-sided error is that $R(\delta) \ket{0}\ket{\pi} = \ket{0}\ket{\pi}$, and $\| (R(\delta)-I\otimes\mathrm{ref}(\pi)) \ket{0}\ket{\psi} \| \leq \delta$ if $\braket{\psi|\pi}=0$.
Note that $R(\delta)$ is a CIQW (Fig.~\ref{fig:QW_model}) with $O\left( \log(\frac{1}{\delta}) \frac{\lambda_{N}}{\lambda_2} \right)$ intermittent steps and $1$ ancillary qubit, from Fig.~\ref{fig:R_delta} and Eq.~\eqref{eq:L_in_R_delta}.

\begin{lemma}\label{lem:approx_ref}
    Suppose that the eigenvector corresponding to the only zero eigenphase of  unitary $U$ is $\ket{\pi}$, and that the gap between the remaining eigenphases and zero is $\lambda $, i.e. $U = \ket{\pi}\bra{\pi} + \sum_{\phi: |\phi|\geq \lambda } e^{i\phi} \ket{\phi}\bra{\phi}$.
    Then the quantum circuit $R(\delta)$ shown in Fig.~\ref{fig:R_delta} uses only one ancillary qubit and $L$ controlled-$U$ and controlled-$U^\dagger$, where $L$ is an odd number such that $L \geq \frac{\ln(4/\delta)}{\lambda/2}$, to approximate the reflection $\mathrm{ref}(\pi) := 2\ket{\pi}\bra{\pi}-I$.
    To be more precise,
    \begin{numcases}{}
        R(\delta) \ket{0}\ket{\pi} = \ket{0}\ket{\pi}; \label{eq:R_0_pi}\\
        \| (R(\delta)+I) \ket{0}\ket{\psi} \| \leq \delta, \ \mathrm{if}\, \braket{\psi|\pi}=0. \label{eq:R_0_psi}
    \end{numcases}
\end{lemma}

\begin{figure}[hbt]
	\centering
	\includegraphics[width=0.7\textwidth]{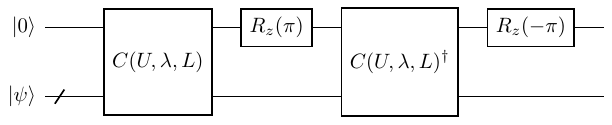}
	\caption{\label{fig:R_delta} The quantum circuit $R(\delta)$ that approximates the reflection $\mathrm{ref}(\pi) = 2\ket{\pi}\bra{\pi}-I$ with one-sided error $\delta$, where the module $C(U,\lambda ,L)$ is shown in Fig.~\ref{res:fig:phase_discrimination}, and $L$ is an odd number such that $L\geq \frac{\ln(4/\delta)}{\lambda/2}$.  }
\end{figure}

\begin{proof}
We first consider the case where the initial state is $\ket{0}\ket{\pi}$.
Since $\ket{\pi}$ is the eigenvector of $U$ whose eigenphase is zero, by Lemma~\ref{res:lem:phase_discrimination}, we have $C(U,\lambda ,L) \ket{0}\ket{\pi} = \ket{0}\ket{\pi}$.
Thus, the effect of $R(\delta)$ on $\ket{0}\ket{\pi}$ can be calculated as follows:
\begin{align}
    R(\delta) \ket{0}\ket{\pi} &= R_z(-\pi) \cdot C(U,\lambda ,L)^\dagger \cdot R_z(\pi)  \ket{0}\ket{\pi} \\
    &= R_z(-\pi) \cdot C(U,\lambda ,L)^\dagger (-i)\ket{0}\ket{\pi} \\
    &= R_z(-\pi) (-i) \ket{0}\ket{\pi} \\
    &= \ket{0}\ket{\pi},
\end{align}
which proves Eq.~\eqref{eq:R_0_pi}.

We then consider the case where the initial state is $\ket{0}\ket{\phi}$ such that $U\ket{\phi} = e^{i\phi} \ket{\phi}$ and$|\phi|\geq \lambda $.
Since $L\geq \frac{\ln(2/(\delta/2))}{\lambda/2}$, by Lemma~\ref{res:lem:phase_discrimination}, it follows that $C(U,\lambda ,L) \ket{0}\ket{\phi} = (a\ket{0}+b\ket{1}) \ket{\phi}$, where $|a|\leq \delta/2$ and $|a|^2+|b|^2=1$.
Thus, the effect of $R(\delta)$ on $\ket{0}\ket{\phi}$ can be calculated as follows:
\begin{align}
	R(\delta) \ket{0}\ket{\phi} &= R_z(-\pi) \cdot C(U,\lambda ,L)^\dagger \cdot R_z(\pi) (a\ket{0}+b\ket{1}) \ket{\phi} \\
	&= R_z(-\pi) \cdot C(U,\lambda ,L)^\dagger (-ia\ket{0}+ib\ket{1}) \ket{\phi} \\
	&= R_z(-\pi) \cdot C(U,\lambda ,L)^\dagger (ia\ket{0}+ib\ket{1} -2ia\ket{0}) \ket{\phi} \\
	&= R_z(-\pi) i \ket{0}\ket{\phi} -R_z(-\pi) \cdot C(U,\lambda ,L)^\dagger 2ia\ket{0} \ket{\phi}\\
	&= -\ket{0}\ket{\phi} +\ket{\delta_\phi}\ket{\phi},
\end{align}
where $\| \ket{\delta_\phi} \| = 2|a| \leq \delta$.

Finally, we consider the case where $\braket{\psi|\pi}=0$.
Then we can expand $\ket{\psi}$ as $\ket{\psi} = \sum_{|\phi|\geq \lambda } c_\phi \ket{\phi}$, where $\sum_{|\phi|\geq \lambda } |c_\phi|^2 = 1$.
Using linearity of $R(\delta)$, we have:
\begin{align}
	R(\delta) \ket{0}\ket{\psi} &= \sum_{|\phi|\geq \lambda } c_\phi R(\delta)\ket{0}\ket{\phi}\\
	&= \sum_{|\phi|\geq \lambda } c_\phi (-\ket{0}\ket{\phi} +\ket{\delta_\phi}\ket{\phi}) \\
	&= -\ket{0} \sum_{|\phi|\geq \lambda } c_\phi \ket{\phi} +\sum_{|\phi|\geq \lambda } c_\phi \ket{\delta_\phi}\ket{\phi}  \\
	&= -\ket{0} \ket{\psi} +\ket{\delta},
\end{align}
where $\| \ket{\delta} \| = \sqrt{\sum_{|\phi|\geq \lambda } |c_\phi|^2 \|\ket{\delta_\phi} \|^2 } \leq \delta \sqrt{\sum_{|\phi|\geq \lambda } |c_\phi|^2} = \delta$.
Thus $\| (R(\delta)+I) \ket{0}\ket{\psi} \| = \| \ket{\delta} \|  \leq \delta$, which proves Eq.~\eqref{eq:R_0_psi}.
We have now finished the proof of Lemma~\ref{lem:approx_ref}.
\end{proof}

Recall that the eigenphases of $U =e^{iLt_0}$, where $t_0 = \frac{\pi}{\lambda_{N}}$, are $0=\frac{\lambda_1\pi}{\lambda_N} <\frac{\lambda_2\pi}{\lambda_{N}} \leq \cdots \leq \frac{\lambda_{N}\pi}{\lambda_{N}} = \pi$.
Thus, the eigenphase gap is $\lambda  = \frac{\pi\lambda_2}{\lambda_{N}}$.
By the above Lemma~\ref{lem:approx_ref}, we can implement the approximate reflection $R(\delta)$ with one-sided error $\delta$ using only one ancillary qubit and $2L$ CTQW controlled-$e^{iLt_0}$, where $L$ is lower bounded by
\begin{equation}\label{eq:L_in_R_delta}
    L \geq \frac{\log(4/\delta)}{\pi\lambda_2/(2\lambda_N)} 
    = O\left(\frac{\lambda_{N}}{\lambda_2} \log(\frac{1}{\delta})\right).
\end{equation}
The total evolution time is
\begin{equation}\label{eq:t_2}
    t_1 := 2Lt_0 = O\left(\frac{1}{\lambda_2} \log(\frac{1}{\delta}) \right).
\end{equation}

\subsection{Search algorithms}\label{subsec:general}
We now prove Theorem~\ref{thm:unknown}.
Recall from Eq.~\eqref{eq:general_oracle_def} the effect of the phase oracle $e^{i\pi \Pi_M}$.
Since we have shown how to implement the one-sided-error approximation $R(\delta)$ of the reflection $\mathrm{ref}(\pi) = 2\ket{\pi}\bra{\pi}-I$ with evolution time $t_1 = O(\frac{1}{\lambda_2} \log(\frac{1}{\delta}) )$,
a straightforward approach to achieve search on general graphs is to combine it with Grover's search algorithm, where $R(\delta)$ is iterated for $O(1/\sqrt{\varepsilon})$ times, assuming the proportion of marked vertices is $\varepsilon = \| \Pi_M \ket{\pi} \|^2$.
To succeed with constant probability, we need to set $\delta = O(\sqrt{\varepsilon})$ in each application of $R(\delta)$, therefore incurring a multiplicative factor $\log(1/\delta)= O(\log(1/\varepsilon))$ in the total cost.

To remove the additional log factor, we will use the recursive amplitude amplification with approximate reflection shown in~\cite{MNRS}.
We first consider the case where the proportion of marked vertices is given, and then deal with the general case where only a lower bound on the proportion of marked vertices is known.

\subsubsection*{Case I: $p_M = \| \Pi_M \ket{\pi} \|^2$ is given.}

\begin{lemma}\label{lem:recur_known}
    Let $t$ be the integer such that $\bar{\varphi}_t := 3^t\arcsin(\sqrt{p_M}) \in[\pi/6,\pi/2]$, where $p_M := \| \Pi_M \ket{\pi} \|^2 >0$ is given.
    Let $\gamma \in(0,1)$.
    Consider the operator $A_t$ defined recursively in Algorithm~\ref{alg:recur_known}.
    Then the final state of $A_t$ applying to $\ket{\pi}\ket{0^t}$ has a constant overlap with the marked vertices:
    \begin{equation}\label{eq:known_success}
        \| \Pi_{M}\otimes I_2^{\otimes t} \cdot A_t \cdot  \ket{\pi}\ket{0^t} \|
        \geq \frac{1}{2} (1-\gamma).
    \end{equation}
    Furthermore, assume that the cost of $\mathrm{ref}(\mathcal{M}^\perp) = e^{i\pi \Pi_M}$ is $c_2$, and the cost of the approximate reflection $R(\beta)$ is $c_1 \log(\frac{1}{\beta})$.
    Then the cost of $A_t$ is
    \begin{equation}\label{eq:known_cost}
        C = O\left(\frac{1}{\sqrt{p_M}}(c_1\log(\frac{1}{\gamma})+c_2) \right).
    \end{equation}
\end{lemma}

Although Lemma~\ref{lem:recur_known} is almost the same as Lemma~1 in Ref.~\cite{MNRS}, the total number of ancillary qubits is reduced from $$O(\log(\lambda_{N}/\lambda_2)\sum_{i=1}^{t} \log(1/\beta_i) ) = O(\log({\lambda_{N}}/{\lambda_2}) \log({1}/{p_M}) \log({\log({1}/{p_M}})/{\gamma}) ),$$ where $t=O(\log(1/p_M))$, to $O(\log(1/p_M))$, since our construction of the approximate reflection $R(\beta)$ shown in Lemma~\ref{lem:approx_ref} uses only one ancillary qubit.
Note also that we let $t$ such that $3^t\arcsin(\sqrt{p_M}) \in[\pi/6,\pi/2]$ instead of $3^t\arcsin(\sqrt{p_M}) \in[\pi/4,3\pi/4]$.
We make this minor modification so that the proof becomes more rigorous (see Remark~\ref{remark:why_pi_4} in Appendix~\ref{app:1} for details), at the small expense of reducing the lower bound on success amplitude from $1/\sqrt{2}-\gamma$ to $1/2-\gamma/2$.
Our revised proof of Lemma~\ref{lem:recur_known} is deferred to Appendix~\ref{app:1}.

\RestyleAlgo{ruled}
\begin{algorithm}[hbt]
    \caption{Recursive operator $A_i$.}\label{alg:recur_known}
    \begin{description}
    \item[1.] Apply $A_{i-1}$ to $\mathcal{H} \otimes \mathcal{K}_1 \otimes \cdots\otimes\mathcal{K}_{i-1}$.
    \item[2.] Apply $\mathrm{ref}(\mathcal{M}^\perp) := e^{i\pi \Pi_M}$ (phase oracle in Eq.~\eqref{eq:general_oracle_def}) to $\mathcal{H}$, which adds phase shift $(-1)$ to marked vertices.
    \item[3.] Apply $A_{i-1}^\dagger$ to $\mathcal{H} \otimes \mathcal{K}_1 \otimes \cdots\otimes\mathcal{K}_{i-1}$.
    
    \item[4.] Apply $R(\beta_i)$ to $\mathcal{H} \otimes \mathcal{K}_i$ conditioned on $\ket{0^{i-1}}$ in registers $\mathcal{K}_1 \otimes \cdots\otimes\mathcal{K}_{i-1}$ with $\beta_i=\frac{9}{2\pi^3}\gamma/i^2$, and apply $(2 \ket{0^{i-1}}\bra{0^{i-1}} - I)$ to $\mathcal{K}_1 \otimes \cdots\otimes\mathcal{K}_{i-1}$.
    \item[5.] Apply $A_{i-1}$ to $\mathcal{H} \otimes \mathcal{K}_1 \otimes \cdots\otimes\mathcal{K}_{i-1}$.
    \end{description}
Note: $A_0 = I$, and each register $\mathcal{K}_i$ consists of $1$ qubit initialized to $\ket{0}$.
\end{algorithm}

\subsubsection*{Case II: a lower bound $\varepsilon$ on $p_M$ is known.}

Using Lemma~\ref{lem:recur_known}, we can show the following Lemma~\ref{lem:recur_unknown}, which is parallel to \cite[Lemma~2]{MNRS}, but the success probability there is lower bounded by $1/12-3\gamma$.
For completeness, we also present the proof of Lemma~\ref{lem:recur_unknown} in Appendix~\ref{app:2}.

\begin{lemma}\label{lem:recur_unknown}
    Let $t_\mathrm{max}$ be the integer such that $3^{t_\mathrm{max}}\arcsin(\sqrt{\varepsilon}) \in [\pi/6,\pi/2]$,
    where $\varepsilon < p_M = \| \Pi_M \ket{\pi} \|^2$.
    Then for $\gamma <\frac{4}{5}(\frac{1}{2} -\frac{\pi}{12}) \approx 0.19$, the search process $S$ shown in Algorithm~\ref{alg:recur_unkonwn} outputs a marked vertex with probability greater than $(\frac{1}{2} -\frac{\pi}{12} -\frac{5\gamma}{4})^2$.
    Furthermore, assume that the cost of $\mathrm{ref}(\mathcal{M}^\perp) = e^{i\pi \Pi_M}$ is $c_2$, and the cost of the approximate reflection $R(\beta)$ is $c_1 \log(\frac{1}{\beta})$.
    Then the cost of $S$ is
    \begin{equation}\label{eq:unknown_cost}
        C = O\left(\frac{1}{\sqrt{\varepsilon}}(c_1\log(\frac{1}{\gamma})+c_2) \right).
    \end{equation}    
\end{lemma}

\RestyleAlgo{ruled,linesnumbered}
\begin{algorithm}[hbt]
\caption{Search process $S$.}\label{alg:recur_unkonwn}
Prepare the initial state $\ket{\pi}\ket{0^{t_\mathrm{max}}} \in \mathcal{H} \otimes \mathcal{K}_1 \otimes \cdots\otimes\mathcal{K}_{t_\mathrm{max}}$\;
\eIf{$t_\mathrm{max}=0$}{Measure $\mathcal{H}$ in the computational basis and output the result\;}{
$i \gets 1$\;
 \While{$i\leq t_\mathrm{max}$}{
  Apply $A_{i}$ to $\mathcal{H} \otimes \mathcal{K}_1 \otimes \cdots\otimes\mathcal{K}_{i}$\;
  Measure $\mathcal{H}$ according to $\Pi_M$, which can be done by applying the standard oracle $O_M$ (Eq.~\eqref{eq:standard_oracle_def}) to $\mathcal{H}$ and one ancillary qubit, and then measuring the ancillary qubit\;
  \If{successful}{
   Measure $\mathcal{H}$ in the computational basis, output the result, and stop\;
   }
  $i \gets i+1$\;
 } 
}
\end{algorithm}


Recall from Eq.~\eqref{eq:t_2} that the evolution time of implementing $R(\beta)$ is $O(\frac{1}{\lambda_2} \log(\frac{1}{\beta}))$.
Thus, the constant $c_1$ in Lemma~\ref{lem:recur_unknown} is $c_1 =O(\frac{1}{\lambda_2})$.
The number $p$ of invocation to the phase oracle $e^{i\pi\Pi_M}$ and the total evolution time $T$ of the CTQW $e^{iLt}$ are as follows:
\begin{equation}
    p = O\left( {1}/{\sqrt{\varepsilon}} \right), \quad
    T = O\left( {1}/(\lambda_2\sqrt{\varepsilon}) \right),
\end{equation}
where $\lambda_2$ is the second-smallest eigenvalue of the graph Laplacian, and $\varepsilon$ is a lower bound on the proportion of marked vertices.
This proves Theorem~\ref{thm:unknown}.

\section{Application to path-finding on graphs}\label{sec:path}

In a recent paper by Li and Zur~\cite{multi_electric}, an ingenious $3$-regular graph based on welded trees is constructed and an exponential speedup of the path-finding problem on this welded tree circuit graph is provided. 
Specifically, they show that a quantum algorithm can solve the path-finding problem with success probability $1-O(\delta)$ and $O(n^{11}\log(n/\delta))$ queries to the adjacency list oracle~\cite[Theorem 5.2]{multi_electric}, while any classical algorithm needs to make $2^{\Omega(n)}$ queries~\cite[Theorem 5.6]{multi_electric}.
Here, $n$ is a parameter that determines the size of the welded tree circuit graph containing $\Theta(n2^n)$ vertices.
We now briefly demonstrate how to improve the query complexity from $O(n^{11}\log(n/\delta))$ to $O(n^8\log(n)\log(n/\delta))$ with QPD (Theorem~\ref{res:lem:phase_discrimination}).

The quantum algorithm for the path-finding problem on the welded tree circuit graph relies on the following Lemma~\ref{lem:phase_filter}, which originates from Lemma~8 in Ref.~\cite{pid19} and Lemma~10 in Ref.~\cite{AP22}.
It tells us that by applying phase estimation~\cite{phase_estimation} to a unitary operator $U$ and an initial state $\ket{\psi}$, and then post-selecting on the $0$-phase, we can approximately project $\ket{\psi}$ to the $1$-eigenspace of the unitary $U$.

\begin{lemma}[Lemma~2.12 in Ref.~\cite{multi_electric}] \label{lem:phase_filter}
    Define the unitary $U_{\cal AB} = (2\Pi_{\cal A} - 1)(2\Pi_{\cal B} - 1)$ acting on a Hilbert space ${\cal H}$ for projectors $\Pi_{\cal A},\Pi_{\cal B}$ onto some subspaces ${\cal A}$ and ${\cal B}$ of ${\cal H}$ respectively.
    Let $\ket{\psi} = \sqrt{p}\ket{\varphi} + (I - \Pi_{\cal A})\ket{\phi}$ be a normalised quantum state such that the normalised vector $\ket{\varphi}$ satisfies $U_{\cal AB}\ket{\varphi} = \ket{\varphi}$ and $\ket{\phi}$ is a (unnormalised) vector satisfying $\Pi_{\cal B}\ket{\phi} = \ket{\phi}$.
    Then performing phase estimation on the state $\ket{\psi}$ with operator $U_{\cal AB}$ and precision $\delta$ outputs ``$0$'' with probability $p' \in [p,p + \frac{17\pi^2\delta\norm{\ket{\phi}}}{16}]$, leaving a state $\ket{\psi'}$ satisfying
    $ \frac{1}{2}\norm{\proj{\psi'} - \proj{\varphi}}_1 \leq \sqrt{\frac{17\pi^2\delta\norm{\ket{\phi}}}{16p}}. $
    Consequently, when the precision is $O\left(\frac{p\epsilon^2}{\norm{\ket{\phi}}}\right)$, the resulting state $\ket{\psi'}$ satisfies
    $ \frac{1}{2}\norm{\proj{\psi'} - \proj{\varphi}}_1 \leq \epsilon. $
\end{lemma}

Using QPD (Theorem~\ref{res:lem:phase_discrimination}), we obtain the following Lemma~\ref{lem:phase_filter_improve} which is an improvement of Lemma~\ref{lem:phase_filter}.
The proof of Lemma~\ref{lem:phase_filter_improve} is deferred to Section~\ref{subsec:phase_filter_proof}.

\begin{lemma}[Improved Lemma~\ref{lem:phase_filter}] \label{lem:phase_filter_improve}
    Let $U_{\cal AB}$, $\ket{\psi}$, $\ket{\varphi}$, $\ket{\phi}$ and $p$ be as in Lemma~\ref{lem:phase_filter}.
    Suppose $p$ has a known lower bound $\bar{p}$ and $\frac{\| \ket{\phi} \|}{\sqrt{p}}$ has a known upper bound $D$.
    For any allowable error $\epsilon\in (0,1)$, let $U= U_{\cal AB}$, $\lambda = \frac{\epsilon}{D}$ and set the odd number $L$ such that $L\geq  \ln(\frac{2\sqrt{2}}{\sqrt{\bar{p}}\epsilon}) \frac{2D} {\epsilon}$.
    Then by applying the quantum circuit $C(U,\lambda,L)$ shown in Fig.~\ref{res:fig:phase_discrimination} to the state $\ket{0}\ket{\psi}$ and then measuring the first qubit, we obtain $\ket{0}$ with probability $p' \in p \times [1,1+\frac{3}{4}\epsilon^2]$, leaving the second register in a state $\ket{\psi'}$ satisfying $\| \ket{\psi'} -\ket{\varphi} \| \leq \epsilon $.
\end{lemma}

From the condition $L\geq  \ln(\frac{2\sqrt{2}}{\sqrt{\bar{p}}\epsilon}) \frac{2D} {\epsilon}$ in Lemma~\ref{lem:phase_filter_improve}, we can see that if a tight lower bound $\bar{p} = \Theta(p)$ on $p$ and a tight upper bound $D =\Theta(\frac{\| \ket{\phi} \|}{\sqrt{p}})$ on $\frac{\| \ket{\phi} \|}{\sqrt{p}}$ are known in advance, the number of calls to the unitary $U_{\cal AB}$ has the following improvement compared to Lemma~\ref{lem:phase_filter}.

\begin{equation}\label{eq:U_calls_improve}
	O\left( \frac{\| \ket{\phi}\|}{p\epsilon^2} \right) \to
	O\left(\log(\frac{1}{\sqrt{p}\epsilon}) \frac{\| \ket{\phi}\|}{\sqrt{p}\epsilon}\right).
\end{equation}

In the quantum algorithm for the path-finding problem on the welded tree circuit graph~\cite[Algorithm~2]{multi_electric}, $1/p = \mathcal{R}_{s,t}^{\rm alt} {\rm w}_s$ is a quantity that can be calculated exactly and of order $\Theta(n^2)$, and $\frac{\| \ket{\phi} \|}{\sqrt{p}} =\| \ket{{\rm p^{alt}}} \|$ is a quantity with an upper bound $D$ that can be calculated exactly and of order $\Theta(n^2)$.
Furthermore, the allowable error $\epsilon$ is set to $\epsilon = \Omega(1/n^2)$, and thus Eq.~\eqref{eq:U_calls_improve} becomes

\begin{equation}\label{eq:U_calls_improve_specific}
	O(n^7) \to O(n^4\log(n^3)).
\end{equation}

To obtain the total query complexity of the quantum algorithm~\cite[Algorithm~2]{multi_electric}, the number of calls to the unitary $U_{\cal AB}$ shown above needs to be multiplied by $\Theta(n^4\log(n/\delta))$, which is the number of repetitions of the phase estimation (discrimination) process to guarantee the overall success probability of $1-O(\delta)$.
See Ref.~\cite{multi_electric} for a more detailed analysis of the quantum algorithm, where a new multidimensional electrical network (with further applications such as traversing the welded tree graph~\cite{CCD03} and the one-dimensional random hierarchical graphs~\cite{BLH23}) is developed by defining Alternative Kirchhoff’s Law and Alternative Ohm’s Law based on the multidimensional quantum walk framework by Jeffery and Zur~\cite{multi}.

\subsection{Proof of Lemma~\ref{lem:phase_filter_improve}}\label{subsec:phase_filter_proof}
To prove Lemma~\ref{lem:phase_filter_improve}, we will also need the following lemma.

\begin{lemma}[Effective spectral gap lemma \cite{state_conversion}] \label{lem:spectral_gap}
    Define the unitary $U_{\cal AB} = (2\Pi_{\cal A} - 1)(2\Pi_{\cal B} - 1)$ acting on a Hilbert space ${\cal H}$ for projectors $\Pi_{\cal A},\Pi_{\cal B}$ onto some subspaces ${\cal A}$ and ${\cal B}$ of ${\cal H}$ respectively.
    Suppose $U_{\cal AB}$ has spectral decomposition $U_{\cal AB} = \sum_{j\in J} e^{i\theta_j} \Pi_j$, where $\theta_j \in (-\pi, \pi]$.
    For any $\epsilon \in [0,\pi)$, let $\Lambda_\epsilon := \sum_{j: |\theta_j|\leq \epsilon} \Pi_j$.
    If $\Pi_{\mathcal{B}} \ket{\phi} = \ket{\phi}$, then
    \begin{equation}
	   \| \Lambda_\epsilon(I-\Pi_{\mathcal{A}})\ket{\phi}\| \leq \frac{\epsilon}{2} \|\ket{\phi}\|.
    \end{equation}
\end{lemma}

\begin{proof}[proof of Lemma~\ref{lem:phase_filter_improve}]
    Suppose $U_{\cal AB}$ has spectral decomposition $U_{\cal AB} = \sum_{j\in J} e^{i\theta_j} \Pi_j$, where $\theta_j \in (-\pi, \pi]$. From $U_{\cal AB}\ket{\varphi} =\ket{\varphi}$, we know:
    \begin{equation}\label{eq:Pi_j_varphi}
    	\Pi_j\ket{\varphi} =
    	\begin{cases}
    		\ket{\varphi}, &{\rm if}\ \theta_j =0;\\
    		0, &{\rm if}\ \theta_j \neq 0.		
    	\end{cases}
    \end{equation}
    Let $\delta = \sqrt{{\bar{p}}/{2}} \epsilon$, then $\delta \leq \sqrt{p/2}\epsilon$ as $\bar{p}$ is a lower bound on $p$.
    Since we have set $\lambda = \frac{\epsilon}{D}$ and $L\geq  \ln(\frac{2\sqrt{2}}{\sqrt{\bar{p}}\epsilon}) \frac{2D} {\epsilon}$, the odd number $L\geq\frac{2\ln(2/\delta)}{\lambda}$ satisfies the condition of Lemma~\ref{res:lem:phase_discrimination}.
    Therefore, when applying the quantum circuit $C(U_{\cal AB},\lambda,L)$ shown in Fig.~\ref{res:fig:phase_discrimination} to the state $\ket{0}\ket{\psi} = \sum_{j\in J} \ket{0} \Pi_j\ket{\psi_j}$, the final state $\sum_{j\in J} \ket{w_j} \Pi_j\ket{\psi}$ satisfies:
    \begin{equation}\label{eq:w_j_zero}
    	\begin{cases}
    		\ket{w_j} =\ket{0}, & {\rm if}\ \theta_j=0; \\
    		\left|\braket{0|w_j} \right| \leq \sqrt{p/2}\epsilon, & {\rm if}\ |\theta_j|\geq \epsilon/D.
    	\end{cases}
    \end{equation}
    We now calculate the probability of obtaining $\ket{0}$ after measuring the first qubit in the computation basis as follows:
    \begin{align}
    	p' &= \| \ket{0}\bra{0}\otimes I \sum_{j\in J} \ket{w_j} \Pi_j\ket{\psi} \|^2 \\
    	&= \| \sum_{j\in J} \braket{0|w_j} \ket{0} \Pi_j\ket{\psi} \|^2 \\
    	&= \sum_{j\in J} \left| \braket{0|w_j} \right|^2  \| \Pi_j\ket{\psi} \|^2\\
    	&= \| \Lambda_0\ket{\psi} \|^2 +\sum_{ j:|\theta_j|\in(0,\frac{\sqrt{p}\epsilon}{\|\ket{\phi}\|}] } \left| \braket{0|w_j} \right|^2  \| \Pi_j\ket{\psi} \|^2
    	+\sum_{ j:|\theta_j| > \frac{\sqrt{p}\epsilon}{\|\ket{\phi}\|} } \left| \braket{0|w_j} \right|^2  \| \Pi_j\ket{\psi} \|^2. \label{eq:p_prime_line4}
    \end{align}
    Consider the first term in Eq.~\eqref{eq:p_prime_line4}. By Lemma~\ref{lem:spectral_gap}, we have:
    \begin{equation}\label{eq:Pi_epsilon_psi}
    	\| \Lambda_\epsilon (\ket{\psi}-\sqrt{p}\ket{\varphi})\| 
    	= \| \Lambda_\epsilon (I-\Pi_{\mathcal{A}}) \ket{\phi} \|
    	\leq \frac{\epsilon}{2} \| \ket{\phi} \|,
    \end{equation}
    which implies $\Lambda_0\ket{\psi} =\sqrt{p}\Lambda_0\ket{\varphi}$ by setting $\epsilon =0$.
    From the fact that  $\Lambda_0 \ket{\varphi} =\ket{\varphi}$ by Eq.~\eqref{eq:Pi_j_varphi} and that $\ket{\varphi}$ is a normalized vector, we have:
    \begin{equation}\label{eq:part_0}
    	\| \Lambda_0\ket{\psi} \|^2 =p.
    \end{equation}
    We can also deduce that
    \begin{equation}\label{eq:varphi_psi}
    	\braket{\varphi|\psi} = \bra{\varphi} \Lambda_0 \ket{\psi} = \sqrt{p} \braket{\varphi|\varphi} =\sqrt{p}.
    \end{equation}
    For the second term in Eq.~\eqref{eq:p_prime_line4}, i.e. $|\theta_j|\in(0,\frac{\sqrt{p}\epsilon}{\|\ket{\phi}\|})$, we have $\| \Pi_j\ket{\psi} \|^2 = \| \Pi_j (\ket{\psi}-\sqrt{p}\ket{\varphi})\|^2 = \| \Pi_j (I-\Pi_{\mathcal{A}}) \ket{\phi} \|^2$, where the first equality follows from Eq.~\eqref{eq:Pi_j_varphi}.
    Together with the fact that $\left|\braket{0|w_j}\right|\leq 1$, we have:
    \begin{align}
    	0 &\leq
    	\sum_{ j:|\theta_j|\in(0,\frac{\sqrt{p}\epsilon}{\|\ket{\phi}\|}] } \left| \braket{0|w_j} \right|^2  \| \Pi_j\ket{\psi} \|^2 \\
    	&\leq \sum_{ j:|\theta_j|\in(0,\frac{\sqrt{p}\epsilon}{\|\ket{\phi}\|}] } \| \Pi_j (I-\Pi_{\mathcal{A}}) \ket{\phi} \|^2 \\
    	&= \| \Lambda_{\frac{\sqrt{p}\epsilon}{\|\ket{\phi}\|}} (I-\Pi_{\mathcal{A}}) \ket{\phi} \|
    	\leq \frac{p\epsilon^2}{4}, \label{eq:part_1_line_3}
    \end{align}
    where the last inequality follows from Eq.~\eqref{eq:Pi_epsilon_psi}.
    For the last term in Eq.~\eqref{eq:p_prime_line4}, we have $|\theta_j| > \frac{\sqrt{p}\epsilon}{\|\ket{\phi}\|} > \epsilon/D$ as $D$ is a upper bound on $\frac{\| \ket{\phi} \|}{\sqrt{p}}$.
    Therefore by the upper bound on $\left| \braket{0|w_j} \right|$ as shown in Eq.~\eqref{eq:w_j_zero}, we have:
    \begin{align}
    	0 &\leq
    	\sum_{ j:|\theta_j| > \frac{\sqrt{p}\epsilon}{\|\ket{\phi}\|} } \left| \braket{0|w_j} \right|^2  \| \Pi_j\ket{\psi} \|^2 \\
    	&\leq \left(\sqrt{p/2}\epsilon\right)^2 \sum_{ j:|\theta_j| > \frac{\sqrt{p}\epsilon}{\|\ket{\phi}\|} } \| \Pi_j\ket{\psi} \|^2
    	\leq \frac{p\epsilon^2}{2}, \label{eq:part_2_line_2}
    \end{align}
    where the last inequality uses the assumption that $\ket{\psi}$ is a normalized vector.
    By substituting Eqs.~\eqref{eq:part_0}, \eqref{eq:part_1_line_3} and \eqref{eq:part_2_line_2} into Eq.~\eqref{eq:p_prime_line4}, we conclude that $p' \in p \times [1,1+\frac{3}{4}\epsilon^2]$.
    This leads to the following lower bound on $p/p'$:
    \begin{equation}\label{eq:1_minus_epsilon}
    	p/p' \geq 1-\frac{3}{4}\epsilon^2
    	\geq 1-\epsilon^2 +\frac{\epsilon^4}{4}
    	= (1-\frac{\epsilon^2}{2})^2,
    \end{equation}
    where the second inequality uses the assumption that $\epsilon \in (0,1)$.
    We can now calculate the inner product between the post-measurement state $\ket{\psi'}$ and the target state $\ket{\varphi}$ as follows:
    \begin{align}
    	\braket{\varphi|\psi'} &=  \bra{0}\bra{\varphi} \cdot \frac{(\ket{0}\bra{0}\otimes I) \cdot C(U,\lambda,L) \cdot \ket{0}\ket{\psi}}{\sqrt{p'}}
    	\\&
    	= \frac{1}{\sqrt{p'}} \braket{\varphi|\psi}
    	=\sqrt{\frac{p}{p'}}, \label{eq:varphi_psi_line_2}
    \end{align}
    where the first equality in Eq.~\eqref{eq:varphi_psi_line_2} uses the fact that $\ket{0}\ket{\varphi}$ remains unchanged under $C(U_{\cal AB},\lambda,L)$, which can be seen from Eqs.~\eqref{eq:Pi_j_varphi} and \eqref{eq:w_j_zero}.
    The second equality in Eq.~\eqref{eq:varphi_psi_line_2} follows from Eq.~\eqref{eq:varphi_psi}.
    Since $\braket{\varphi|\psi'} \in \mathbb{R}$, we have:
    \begin{align}
    	\| \ket{\psi'} -\ket{\varphi} \|^2 &= 2\left(1-\sqrt{\frac{p}{p'}}\right) \leq \epsilon^2,
    \end{align}
    where the last inequality follows from Eq.~\eqref{eq:1_minus_epsilon}.
\end{proof}

\section{Conclusion and discussion}\label{sec:discussion}
In this paper, we studied the phase discrimination problem of deciding whether the eigenphase of a given eigenstate of a unitary $U$ is zero or not, and proposed a quantum algorithm called quantum phase discrimination (QPD) for this problem with optimal query complexity to controlled-$U$ and only one ancillary qubit.
The angles of single qubit $Y$ rotations acting on the ancillary qubit are given by a simple analytical formula, compared with other straightforward methods based on quantum signal processing (QSP) and its variations.
In light of the problem of how to bypass the computationally intensive angle-finding step in QSP has also been considered recently by Alase~\cite{without_angle}, it will be interesting to find analytical angle parameters that approximate functions other than the ``delta'' function considered in this paper.

We also discussed applications of QPD, focusing on two problems related to quantum search on graphs: (i) For the fundamental problem of spatial search on graphs, we propose the controlled intermittent quantum walks (CIQW) model inspired by the structure of QPD, and in combination with QPD, we obtain a novel quantum search algorithm; (ii) For a path-finding problem on a welded tree circuit graph by Li and Zur~\cite{multi_electric}, we improve the quantum algorithm's query complexity using QPD.
It will be exciting to find more quantum algorithms that can benefit from QPD.

\bibliographystyle{quantum}
\bibliography{ref}

\appendix

\section{The welded tree circuit graph}\label{app:circuit_graph}
The welded tree circuit graph $G$ constructed by Li and Zur is shown in Fig.~\ref{fig:welded_circuit_graph_simple}.
It consists of $n$ isomorphic graphs $G_i$ connected one by one, and each subgraph $G_i$ contains three welded tree of depth $n$ as depicted by the three diamonds $W_1,W_2,W_3$ in Fig.~\ref{fig:welded_circuit_graph_simple}.
The only vertex with degree $2$ is the starting vertex $s$, and the only vertex with degree $1$ is the target vertex $t$, and all the other vertices have degree $3$.
The goal of the path-finding problem is to find a path in $G$ that connects $s$ and $t$.

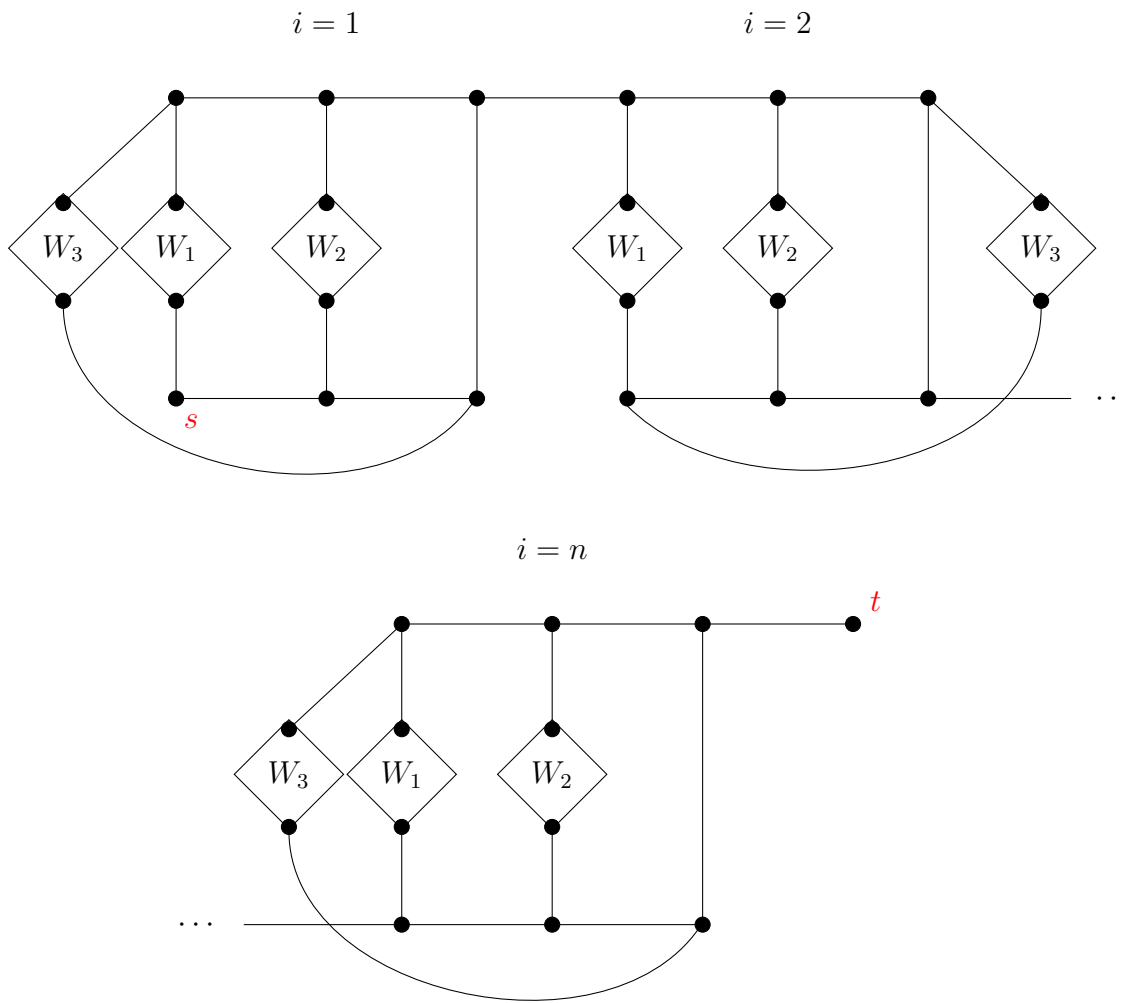
\begin{figure}[!ht]
\centering
\begin{tikzpicture}

\node at (2,5) {$i=1$};
\filldraw (0,0) circle (.1);     \node at (0.2,-0.3) {{\color{red}$s$}};
\filldraw (2,0) circle (.1);          \node at (2,-0.3) {};
\filldraw (4,0) circle (.1);     \node at (4.6,0) {};
\filldraw (0,4) circle (.1);          \node at (0,4.3) {};
\filldraw (2,4) circle (.1);     \node at (2,4.3) {};
\filldraw (4,4) circle (.1);          \node at (4,4.3) {};

\node[diamond,
  draw = black,
  minimum width = 1cm,
  minimum height = 1.3cm] (d) at (0,2) {$W_1$};
\node[diamond,
  draw = black,
  minimum width = 1cm,
  minimum height = 1.3cm] (d) at (2,2) {$W_2$};
\node[diamond,
  draw = black,
  minimum width = 1cm,
  minimum height = 1.3cm] (d) at (-1.5,2) {$W_3$};

\filldraw (0,1.3) circle (.1);         \node at (0.7,1.3) {};
\filldraw (2,1.3) circle (.1);    \node at (3,1.3) {};
\filldraw (0,2.6) circle (.1);    \node at (1,2.6) {};
\filldraw (2,2.6) circle (.1);         \node at (2.9,2.6) {};
\filldraw (-1.5,1.3) circle (.1);        \node at (-2.6,1.3) {};
\filldraw (-1.5,2.6) circle (.1);  \node at (-2.6,2.6) {};

\draw[-] (0.1,0)--(1.9,0);       \node at (1,0.3) {};
\draw[-] (0,0.1)--(0,1.2);       \node at (0.2,0.7) {};
\draw[-] (2.1,0)--(3.9,0);       \node at (3,0.3) {};
\draw[-] (2,0.1)--(2,1.2);       \node at (2.2,0.7) {};
\draw[-] (4,3.9)--(4,0.1);       \node at (4.3,2) {};
\draw[-] (0,2.7)--(0,3.9);       \node at (0.2,3.3) {};
\draw[-] (2,2.7)--(2,3.9);       \node at (2.2,3.3) {};
\draw[-] (-.05,3.95)--(-1.45,2.65);  \node at (-1,3.5) {};
\draw[-] (0.1,4)--(1.9,4);       \node at (1,4.3) {};
\draw[-] (3.9,4)--(2.1,4);       \node at (3,4.3) {};
\draw[-] (5.9,4)--(4.1,4);       \node at (5,4.3) {};
\draw[-] (-1.5,1.25) to[out=-90,in=-125] (3.95,-0.05); \node at (3.5,-1) {};

\node at (8,5) {$i=2$};
\filldraw (6,0) circle (.1);     \node at (5.7,-0.3) {};
\filldraw (8,0) circle (.1);          \node at (8,-0.3) {};
\filldraw (10,0) circle (.1);     \node at (10.6,0.3) {};
\filldraw (6,4) circle (.1);          \node at (6,4.3) {};
\filldraw (8,4) circle (.1);     \node at (8,4.3) {};
\filldraw (10,4) circle (.1);          \node at (10,4.3) {};

\node[diamond,
  draw = black,
  minimum width = 1cm,
  minimum height = 1.3cm] (d) at (6,2) {$W_1$};
\node[diamond,
  draw = black,
  minimum width = 1cm,
  minimum height = 1.3cm] (d) at (8,2) {$W_2$};
\node[diamond,
  draw = black,
  minimum width = 1cm,
  minimum height = 1.3cm] (d) at (11.5,2) {$W_3$};

\filldraw (6,1.3) circle (.1);         \node at (7,1.3) {};
\filldraw (8,1.3) circle (.1);    \node at (8.8,1.3) {};
\filldraw (6,2.6) circle (.1);    \node at (6.7,2.6) {};
\filldraw (8,2.6) circle (.1);         \node at (9,2.6) {};
\filldraw (11.5,1.3) circle (.1);        \node at (12.6,1.3) {};
\filldraw (11.5,2.6) circle (.1);  \node at (12.3,2.6) {};

\draw[-] (6.1,0)--(7.9,0);     \node at (7,0.3) {};
\draw[-] (6,1.2)--(6,0.1);     \node at (6.2,0.7) {};
\draw[-] (9.9,0)--(8.1,0);     \node at (9,0.3) {};
\draw[-] (8,1.2)--(8,0.1);     \node at (8.3,0.7) {};
\draw[-] (10,0.1)--(10,3.9);    \node at (10.3,2) {};
\draw[-] (6,3.9)--(6,2.7);     \node at (6.2,3.3) {};
\draw[-] (8,3.9)--(8,2.7);     \node at (8.2,3.3) {};
\draw[-] (11.45,2.65)--(10.05,3.95);\node at (11,3.5) {};
\draw[-] (6.1,4)--(7.9,4);     \node at (7,4.3) {};
\draw[-] (8.1,4)--(9.9,4);     \node at (9,4.3) {};
\draw[-] (11.9,0)--(10.1,0);    \node at (11.5,-0.3) {};
\draw[-] (6,-0.1) to[out=-45,in=-90] (11.5,1.2); \node at (6.5,-1) {};

\node at (12.5,0) {$\cdots$};

\node at (5,-2) {$i=n$};
\filldraw (3,-7) circle (.1);     \node at (3,-7.3) {};
\filldraw (5,-7) circle (.1);          \node at (5,-7.3) {};
\filldraw (7,-7) circle (.1);     \node at (7.6,-7) {};
\filldraw (3,-3) circle (.1);          \node at (3,-2.7) {};
\filldraw (5,-3) circle (.1);     \node at (5,-2.7) {};
\filldraw (7,-3) circle (.1);          \node at (7,-2.7) {};
\filldraw (9,-3) circle (.1);          \node at (9.3,-2.7) {};
\draw[-] (9,-3)--(7.1,-3);       \node at (8,-2.7) {};

\node[diamond,
  draw = black,
  minimum width = 1cm,
  minimum height = 1.3cm] (d) at (3,-5) {$W_1$};
\node[diamond,
  draw = black,
  minimum width = 1cm,
  minimum height = 1.3cm] (d) at (5,-5) {$W_2$};
\node[diamond,
  draw = black,
  minimum width = 1cm,
  minimum height = 1.3cm] (d) at (1.5,-5) {$W_3$};

\filldraw (3,-5.7) circle (.1);         \node at (3.7,-5.7) {};

\node at (9.3,-2.7) {{\color{red} $t$}};

\filldraw (5,-5.7) circle (.1);    \node at (6,-5.7) {};
\filldraw (3,-4.4) circle (.1);    \node at (4,-4.4) {};
\filldraw (5,-4.4) circle (.1);         \node at (5.9,-4.4) {};
\filldraw (1.5,-5.7) circle (.1);        \node at (0.4,-5.7) {};
\filldraw (1.5,-4.4) circle (.1);  \node at (0.4,-4.4) {};

\draw[-] (3.1,-7)--(4.9,-7);       \node at (4,-6.7) {};
\draw[-] (3,-6.9)--(3,-5.8);       \node at (3.2,-6.3) {};
\draw[-] (5.1,-7)--(6.9,-7);       \node at (6,-6.7) {};
\draw[-] (5,-6.9)--(5,-5.8);       \node at (5.2,-6.3) {};
\draw[-] (7,-3.1)--(7,-6.9);       \node at (7.3,-5) {};
\draw[-] (3,-4.3)--(3,-3.1);       \node at (3.2,-3.7) {};
\draw[-] (5,-4.3)--(5,-3.1);       \node at (5.2,-3.7) {};
\draw[-] (2.95,-3.05)--(1.55,-4.35);  \node at (2,-3.5) {};
\draw[-] (3.1,-3)--(4.9,-3);       \node at (4,-2.7) {};
\draw[-] (6.9,-3)--(5.1,-3);       \node at (6,-2.7) {};
\draw[-] (2.9,-7)--(0.9,-7);       \node at (1.5,-7.3) {};
\draw[-] (1.5,-5.75) to[out=-90,in=-125] (6.95,-7.05); \node at (6.5,-8) {};

\node at (0.3,-7) {$\cdots$};

\end{tikzpicture}
\caption{The welded tree circuit graph~\cite[Figure~1]{multi_electric}.}\label{fig:welded_circuit_graph_simple}
\end{figure}

\section{Other trivial methods for phase discrimination}\label{app:other_methods}

\textbf{Solving phase discrimination with quantum phase processing (QPP).} Perhaps the most relevant work to our QPD is the quantum phase processing (QPP) proposed by Wang et al.~\cite{phase_processing}. 
Their main Theorem~1 implies that for any trigonometric polynomial $F(\phi) = \sum_{j=-L}^{L}c_j e^{ij\phi}$, with $|F(\phi)|\leq 1$ for all $\phi\in\mathbb{R}$, there exists a quantum circuit $V(U,\vec{\alpha},\vec{\beta})$ with $1$ ancillary qubit and consisting of $L$ controlled-$U$, $L$ controlled-$U^\dagger$, $2L+2$ single-qubit $z$ rotations $R_z(\alpha_i)$ and $2L+1$ single-qubit $y$ rotations $R_y(\beta_i)$ (see Fig.~1 in \cite{phase_processing}), such that
\begin{equation}
	\bra{0}\otimes I \cdot V(U,\vec{\alpha},\vec{\beta})\cdot \ket{0}\ket{\psi} = F(\phi),
\end{equation}
where $U\ket{\psi} = e^{i\phi}\ket{\psi}$.
To solve the phase discrimination problem using QPP, we simply let
\begin{equation}\label{eq:F_phi}
    F(\phi) := \frac{T_{2L}\left({\cos(\frac{\phi}{2})}/{\cos(\frac{\lambda }{2})}\right)} {T_{2L}({1}/{\cos(\frac{\lambda }{2})})},
\end{equation}
which satisfies $F(0)=1$, and it is a trigonometric polynomial $\sum_{j=-L}^{L}c_j e^{ij\phi}$ with coefficients given by \cite[Eq.~(3)]{dolph_optimal_filter}
\begin{equation}
    c_j = \frac{1}{2L+1} [1+2 \sum_{m=1}^{L} \cos(\frac{2\pi mj}{2L+1}) \cdot F(\frac{2\pi m}{2L+1})].
\end{equation}
If $2L \geq \frac{\ln(2/\delta)}{\lambda/2}$, then $|F(\phi)|\leq\delta$ when $|\phi| \in[\lambda,\pi]$ (see Eq.~\eqref{eq:0_omega_delta} in Theorem~\ref{res:lem:phase_discrimination}), and thus the ancillary qubit indicates whether the eigenphase $\phi$ is zero or not with one-sided error.
Note, however, that controlled-$U^\dagger$ is needed and the angles $\vec{\theta} \in \mathbb{R}^{2L+2}, \vec{\phi}\in\mathbb{R}^{2L+1}$ are computed numerically and iteratively~\cite[Algorithm 3]{phase_processing}.

\textbf{Solving phase discrimination with eigenstate filtering.} One can also use the eigenstate filtering method~\cite{QLSS_20} to solve the phase discrimination problem.
However, this would require a non-trivial conversion (Appendix~\ref{app:filtering}) since eigenstate filtering and phase discrimination are two different problems, with the former dealing with block-encoded Hermitian matrices.
In addition, the eigenstate filtering method requires a numerical computation of the parameters in the gate sequence.

\subsection{Quantum phase discrimination using eigenstate filtering}\label{app:filtering}

We now show that the phase discrimination problem (Definition~\ref{def:phase}) can also be solved using the eigenstate filtering method of Lin and Tong~\cite{QLSS_20} with a non-trivial conversion.


A unitary matrix $U$ is said to be an $a$-qubit block encoding of a square matrix $A$, if $A = (\bra{0^a}\otimes I)U(\ket{0^a}\otimes I)$, or graphically:
\begin{equation}
	U=\begin{bmatrix}
		A & * \\
		* & *
	\end{bmatrix}.
\end{equation}

The eigenstate filtering method relies on the following result of the polynomial eigenvalue transformation via quantum signal processing~\cite[Theorem 1']{QLSS_20}, which follows from Ref.~\cite[Theorem 2]{qsvt}.
The result works for polynomials with definite parity and general square matrix $A$, but we will only use even polynomials and Hermitian matrix $H$ as stated below.

\begin{lemma}\label{lem:qsvt}
    Suppose $U_H$ is an $a$-qubit block encoding of a Hermitian matrix $H = \sum_{j=1}^{N} \lambda_j \ket{\psi_j}\bra{\psi_j}$, where $\lambda_j \in[-1,1]$.
    Consider an even polynomial $P: [-1,1] \to [-1,1]$ of degree $d$, then there exist parameters $\theta_k\in \mathbb{R}$, $k\in\{1,\dots,d\}$ that can be classically and approximately computed such that the quantum circuit shown in Fig.~\ref{fig:qsvt} is an $(a+1)$-qubit block encoding of $P(H) = \sum_{i=1}^{N} P(\lambda_j) \ket{\psi_j}\bra{\psi_j}$.
\end{lemma}

\begin{figure}[hbt]
	\centering
	\includegraphics[width=\textwidth]{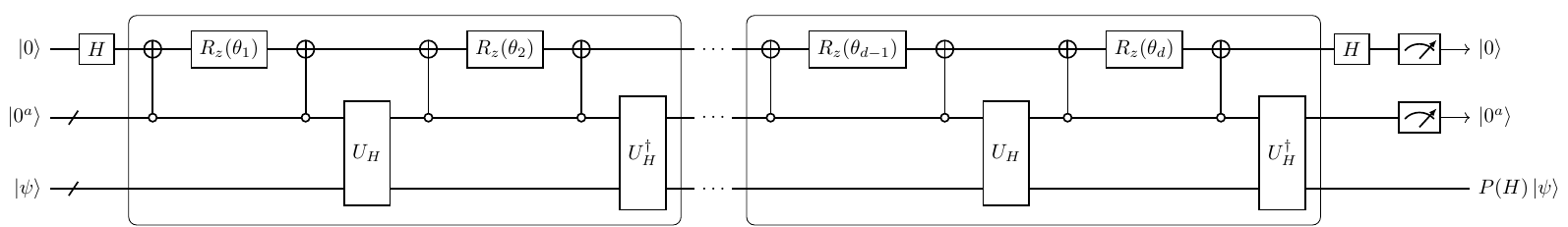}
	\caption{\label{fig:qsvt} Quantum circuit implementing an $(a+1)$-qubit block encoding of $P(H)$, where $P: [-1,1] \to [-1,1]$ is an even polynomial of degree $d$, and $U_H$ is an $a$-qubit block encoding of $H$ (see Fig.~\ref{fig:encode_U} for the quantum circuit of $U_H$ that we will use for QPD). }
\end{figure}

To filter out the non-zero eigenvalue of the Hermitian matrix $H$, Lin and Tong~\cite{QLSS_20} design the following even polynomial $R_l(x;\Delta)$ of degree $2l$, where $\Delta \in(0,1)$ is the gap between the zero and non-zero eigenvalues of $H$.
\begin{equation}\label{eq:R_l_Delta_def}
	R_l(x;\Delta) := \frac{T_l\left(-1+2\frac{x^2-\Delta^2}{1-\Delta^2} \right)}{T_l\left(-1+2\frac{-\Delta^2}{1-\Delta^2} \right)},
\end{equation}
where $T_l(x)$ is the Chebyshev polynomial of the first kind.
It is shown in Ref.~\cite[Lemma 2]{QLSS_20} that $R_l(x;\Delta)$ has the desired properties of filtering out the non-zero eigenvalues as stated in the following lemma.
Note that property 2 is an improvement on the result of ``$\left| R_l(x;\Delta) \right| \leq 2e^{-\sqrt{2}l\Delta}$ and $\Delta\in(0,1/\sqrt{12}]$'' shown in Ref.~\cite[Lemma 2]{QLSS_20}, for which we provide below an alternative proof compared to Ref.~\cite[Lemma 13]{QLSS_20}.
\begin{lemma}\label{lem:R_l_properties}
    The even polynomial $R_l(x;\Delta)$ of degree $2l$ has the following two properties:
    \begin{enumerate}
        \item $\left| R_l(x;\Delta) \right| \leq 1$ for all $\left|x\right| \leq 1$, and $R_l(0;\Delta) = 1$.
        \item $\left| R_l(x;\Delta) \right| \leq 2e^{-2l\arcsin(\Delta)}$ when $\left|x\right| \in [\Delta,1]$.
    \end{enumerate}
\end{lemma}

\begin{proof}
Using the parity of Chebyshev polynomials that $T_l(-x)=(-1)^l T_l(x)$, and the nesting property that $T_m(T_n(x))=T_{mn}(x)$, and the explicit expression of $T_2(x)=2x^2-1$, the even polynomial $R_l(x;\Delta)$ of degree $2l$ can be transformed as follows:
\begin{equation}\label{eq:R_l_x_transform}
	R_l(x;\Delta) 
	= \frac{T_l\left(1-2\frac{x^2-\Delta^2}{1-\Delta^2} \right)}{T_l\left(1-2\frac{-\Delta^2}{1-\Delta^2} \right)}
	= \frac{T_l\left(2\frac{1-x^2}{1-\Delta^2}-1 \right)}{T_l\left(2\frac{1}{1-\Delta^2}-1 \right)}
	= \frac{T_{2l}\left(\sqrt{\frac{1-x^2}{1-\Delta^2}}\right)}{T_{2l}\left(\frac{1}{\sqrt{1-\Delta^2}}\right)}.
\end{equation}

When $\left|x\right| \in [0,\Delta]$, it is easy to see that $\sqrt{\frac{1-x^2}{1-\Delta^2}} \in \left[1, \frac{1}{\sqrt{1-\Delta^2}}\right]$.
Since $T_{2l}(1)=1$ and $T_{2l}(x)$ is strictly monotonically increasing for $x \geq 1$, the numerator $T_{2l}\left(\sqrt{\frac{1-x^2}{1-\Delta^2}}\right) \in \left[1, T_{2l}\left(\frac{1}{\sqrt{1-\Delta^2}}\right)\right]$ and the denominator is greater than $1$.

When $\left|x\right| \in [\Delta,1]$, using the fact that $\left| T_l(x) \right| \leq 1$ for all $\left| x \right| \leq 1$, we can see that the absolute value of the numerator $T_{2l}\left(\sqrt{\frac{1-x^2}{1-\Delta^2}}\right)$ does not exceed $1$.
Since $R_l(0;\Delta) = 1$ trivially holds, we have shown property 1.
To show property 2, it suffices to prove that the denominator is greater than $\frac{1}{2} e^{2l\arcsin(\Delta)}$ as shown below.

From the argument below Eq.~\eqref{eq:L_lower_relax}, it can be seen that ${\rm arccosh}(1/\cos(\theta)) \geq \theta$ for $\theta\in[0,\pi/2)$.
Using the fact that $\cosh(x) = (e^x+e^{-x})/2 > e^x/2$ for $x\in \mathbb{R}$, and the fact that $\cosh(x)$ is strictly monotonically increasing for $x \geq 0$, and the definition of $T_l(x) = \cosh(l\cdot{\rm arccosh}(x))$ for $|x| \geq 1$, it is easy to see that $T_l(1/\cos(\theta)) > \frac{1}{2}e^{l\theta}$ for $\theta\in[0,\pi/2)$.
Therefore, the denominator $T_{2l}\left(\frac{1}{\sqrt{1-\Delta^2}}\right) > \frac{1}{2} e^{2l\arcsin(\Delta)}$ for $\Delta \in [0,1)$.
\end{proof}

We propose the following lemma that gives a way to map the eigenphase $\phi$ of $U\ket{\psi} = e^{i\phi} \ket{\psi}$ to the eigenvalue $\sin(\phi)$ of a block encoded Hermitian $H$ such that $H\ket{\psi} = \sin(\phi) \ket{\psi}$.

\begin{lemma}\label{lem:encode_H}
Suppose the unitary matrix $U$ in the phase discrimination problem has the following spectral decomposition:
\begin{equation}
    U =\sum_{j=1}^{N} e^{i\phi_j}\ket{\psi_j}\bra{\psi_j},
\end{equation}
then the quantum circuit $U_H$ shown in Fig.~\ref{fig:encode_U} is a $1$-qubit block encoding of the Hermitian matrix:
\begin{equation}\label{eq:encode_H}
    H :=\sum_{i=j}^{N} \sin(\phi_j) \ket{\psi_j}\bra{\psi_j}.
\end{equation}
\end{lemma}

\begin{figure}[hbt]
	\centering
	\includegraphics[width=0.6\textwidth]{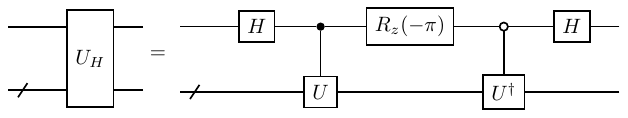}
	\caption{\label{fig:encode_U} A quantum circuit $U_H$ that implements a $1$-qubit block encoding of the Hermitian matrix $H$ (Eq.~\eqref{eq:encode_H}), where $H$ with a square represents a Hadamard gate and $R_z(\theta) ={\rm diag}(e^{-i\theta/2},e^{i\theta/2})$. }
\end{figure}

\begin{proof}
The matrix expression of the quantum circuit $U_H$ shown in Fig.~\ref{fig:encode_U} can be calculated as follows:
\begin{align}
	U_H &=\frac{1}{2}\begin{bmatrix}
		I & I \\
		I & -I
	\end{bmatrix}
	\cdot\begin{bmatrix}
		U^\dagger & 0 \\
		0 & I
	\end{bmatrix}
	\cdot\begin{bmatrix}
		iI & 0 \\
		0 & -iI
	\end{bmatrix}
	\cdot\begin{bmatrix}
		I & 0 \\
		0 & U
	\end{bmatrix}
	\cdot\begin{bmatrix}
		I & I \\
		I & -I
	\end{bmatrix}\\    
    &=\frac{1}{2}\begin{bmatrix}
		i(U^\dagger-U) & i(U+U^\dagger)\\
		i(U+U^\dagger) & i(U^\dagger-U)
	\end{bmatrix}.
\end{align}
By the Euler formula $\sin(\phi) = \frac{i}{2}(e^{-i\phi} -e^{i\phi})$, we have $\frac{i}{2}(U^\dagger-U) =\sum_{i=j}^{N} \sin(\phi_j) \ket{\psi_j}\bra{\psi_j}$, and thus $U_H$ is a $1$-qubit block encoding of the Hermitian matrix $H$ (Eq.~\eqref{eq:encode_H}).
\end{proof}

We can now show the following theorem that solves the phase discrimination problem using the eigenstate filtering method.

\begin{theorem}
    Suppose $U\ket{\psi} =e^{i\phi} \ket{\psi}$, where $\phi \in (-\pi, \pi]$.
    Replace $U_H$ in the quantum circuit shown in Fig.~\ref{fig:qsvt} with the quantum circuit shown in Fig.~\ref{fig:encode_U}.
    Suppose even number $d \geq \frac{\ln(2/\delta)}{\lambda}$, then there exist parameters $\theta_k\in \mathbb{R}$, $k\in\{1,\dots,d\}$ that can be classically and approximately computed such that the probability $p$ of obtaining $\ket{0^2}$ when measuring the first two ancillary qubits satisfies: $p=1$ when $\phi =0$; and $p\leq \delta^2$ when $|\phi| \in [\lambda,\pi]$.
\end{theorem}

\begin{proof}
Write the even number $d$ as $d=2l$, then $l \geq \frac{\ln(2/\delta)}{2\lambda}$.
Let $\Delta := \sin(\lambda)$.
Suppose $|\phi| \in [\lambda, \pi]$, then $\left| \sin(\phi) \right| \in [\Delta,1]$.
From property 2 of $R_l(x;\Delta)$ shown in Lemma~\ref{lem:R_l_properties}, we have $\left|R_l(\sin(\phi);\Delta)\right| \leq 2 e^{-2l\lambda} \leq \delta$.
Since $R_l(0;\Delta) =1$ by property 1 in Lemma~\ref{lem:R_l_properties}, we have $\left|R_l(\sin(\phi);\Delta)\right| = 1$ when $\phi = 0$.

Let $P(x) := R_l(x;\Delta)$.
From property 1 of $R_l(x;\Delta)$ shown in Lemma~\ref{lem:R_l_properties}, we know $P(x) : [-1,1] \to [-1,1]$ and it is an even polynomial of degree $d=2l$.

From Lemma~\ref{lem:encode_H}, we know the quantum circuit $U_H$ shown in Fig.~\ref{fig:encode_U} is a $1$-qubit block encoding of $H =\sum_{i=j}^{N} \sin(\phi_j) \ket{\psi_j}\bra{\psi_j}$.

By Lemma~\ref{lem:qsvt}, there exist parameters $\theta_k\in \mathbb{R}$, $k\in\{1,\dots,d\}$ that can be classically and approximately computed such that the quantum circuit shown in Fig.~\ref{fig:qsvt} is a $2$-qubit block encoding of $P(H) = R_l(H; \Delta) = \sum_{j=1}^{N} R_l(\sin(\phi_j);\Delta) \ket{\psi_j}\bra{\psi_j}$.
Thus, the probability $p$ of obtaining $\ket{0^2}$ when measuring the first two ancillary qubits of the quantum circuit shown in Fig.~\ref{fig:qsvt} is
\begin{equation}
    p = \| P(H) \ket{\psi} \|^2 = |R_l(\sin(\phi);\Delta)|^2,
\end{equation}
which satisfies: $p=1$ when $\phi =0$; and $p\leq \delta^2$ when $|\phi| \in [\lambda,\pi]$.
\end{proof}

\section{Hitting time of the $n$-cycle graph}\label{app:n_cycle}

Suppose the vertex set of the $n$-cycle graph is $V=\{0,1,\dots,n-1\}$.
Then its adjacency matrix $A$ can be expressed as:
\begin{equation}
    A = \sum_{j=0}^{n-1} \ket{j} \big(\bra{j-1}+\bra{j+1} \big),
\end{equation}
where the addition operation is modulo $n$, i.e. $-1 \equiv n-1 ({\rm mod}\ n)$ and $n \equiv 0 ({\rm mod}\ n)$.
Let $\ket{v_k} := \sum_{j=0}^{n-1} \omega^{jk} \ket{j}$, where $\omega := \exp(\frac{2\pi i}{n})$.
Then it can be verified that $\ket{v_k}$ is the eigenvalue of $A$ whose eigenvalue is $2 \cos(\frac{2\pi k}{n})$.
Since $\braket{v_k | v_{k'}} =0$ for different $k,k' \in \{0,1,\dots,n-1\}$, the set of eigenvalues of $A$ is $\{2 \cos(\frac{2\pi k}{n}) : k\in\{0,\dots,n-1\} \}$.

Since the degree of each vertex in the $n$-cycle graph is equal to $2$, assuming that the transition probability from every vertex to its two adjacent vertices are both $\frac{1}{2}$, then the Markov chain on the $n$-cycle graph is:
\begin{equation}
    P = \frac{1}{2} A.
\end{equation}
Therefore, the spectral gap $\delta$ of $P$, i.e. the minimum distance between the $1$ eigenvalue and the remaining eigenvalues, is as follows:
\begin{equation}
	\delta = 1-\cos(\frac{2\pi}{n}) = 2\sin^2(\frac{\pi}{n}) = \Theta(\frac{1}{n^2}).
\end{equation}

Let $M \subseteq V$ be the set of marked vertices.
Suppose there is only one marked vertex on the $n$-cycle, then without loss of generality, we can assume $M = \{0\}$.
Let $\ket{\pi}$ be the stationary distribution of $P$, which is $\ket{\pi} = \frac{1}{n} \sum_{j=0}^{n-1} \ket{j}$ on the $n$-cycle graph.
Let $\ket{\pi_M}$ be the vector obtained from $\ket{\pi}$ by deleting entries indexed from $M$.
Here, $\ket{\pi_M} = \frac{1}{n} \sum_{j=1}^{n-1} \ket{j}$.
Let $P_M$ be the matrix obtained from $P$ by deleting its rows and columns indexed from $M$, which is $P_M = \frac{1}{2} \sum_{j=1}^{n-2} \ket{j}\bra{j+1} + \ket{j+1}\bra{j}$ on the $n$-cycle graph.
Let $\ket{\vec{1}} := \sum_{j=1}^{n-1} \ket{j}$ be the vector whose all coordinates equal to one.
The hitting time of $P$ can be calculated as follows~\cite[Eq. (1)]{Szegedy_03}.
\begin{equation}
	{\rm HT} = \bra{\pi_M} (I-P_M)^{-1} \ket{\vec{1}}.
\end{equation}

By the inverse formula of a tridiagonal Toeplitz matrix~\cite{Toeplitz}, we have:
\begin{equation}
	(I-P_M)^{-1} =
    \frac{2}{n} \sum_{i=1}^{n-1} \sum_{j=i}^{n-1} i(n-j)  \ket{i}\bra{j} 
    + \frac{2}{n} \sum_{i=2}^{n-1} \sum_{j=1}^{i} (n-i)j \ket{i}\bra{j}.
\end{equation}
Thus the hitting time of the $n$-cycle graph with one marked vertex can be calculated as follows:
\begin{align}
	{\rm HT} &= \frac{1}{n} \left[ \frac{2}{n} \sum_{i=1}^{n-1} \sum_{j=i}^{n-1} i(n-j)   
    + \frac{2}{n} \sum_{i=2}^{n-1} \sum_{j=1}^{i} (n-i)j \right] \\
    &= \frac{1}{n^2} \left[ \sum_{i=1}^{n-1}i (n+1-i)(n+i)   
    + \sum_{i=2}^{n-1} (n-i) (i+1)i \right] \\
    &= \frac{1}{n^2} \left[ - 2(n-1) + \sum_{i=1}^{n-1} n(n+2)i + ni^2 -2i^3 \right] \\
    &= \frac{1}{n^2} \left[ - 2(n-1) + n(n+2)\frac{n(n-1)}{2} + n\frac{(n-1)n(2n-1)}{6} -2\frac{(n-1)^2n^2}{4} \right] \\
    &= \frac{n^2}{3} + n - \frac{4}{3} - \frac{2}{n} + \frac{2}{n^2} \\
    &= \Theta(n^2).
\end{align}

\section{Proof of Lemma~\ref{lem:recur_known} }\label{app:1}

Let $\bar{\varphi}_0 := \arcsin(\sqrt{p_M}) \in(0,\pi/2]$, and $\bar{\varphi}_i := 3^i \bar{\varphi}_0$.
Since $\cup_{k\geq 0} [\pi/6,\pi/2]/3^k =(0,\pi/2]$, there exists an integer $t\geq 0$ such that $\bar{\varphi}_t = 3^t\bar{\varphi}_0 \in [\pi/6,\pi/2]$.
Note that Lemma~\ref{lem:recur_known} holds trivially when $t=0$, in which case we simply measure the initial state $\ket{\pi}$ and $\| \Pi_M \ket{\pi} \| =\sqrt{p_M} \geq \sin(\frac{\pi}{6}) =\frac{1}{2}$.
We can therefore assume $t\geq 1$, or equivalently $\bar{\varphi}_0 \leq \pi/6$.

Denote by $\ket{\phi_i} := A_i\ket{\pi}\ket{0^i}$ the middle state of applying $A_i$ to the initial state $\ket{\pi}\ket{0^i}$.
Let $\sin(\varphi_i) := \| \Pi_M\otimes I_2^{\otimes i} \ket{\phi_i} \|$ be the
length of the projection of $\ket{\phi_i}$ to the subspace spanned by the marked vertices.
We will later prove the following claim.

\begin{claim}\label{claim:e_i_upper}
Denote by $e_i := |\sin(\varphi_i) -\sin(\bar{\varphi}_i)|$ the deviation from the idea success amplitude $\sin(\bar{\varphi}_i)$.
Let $\tilde{e}_i$ be defined by the following recurrence relation:
\begin{align}
    \tilde{e}_0 &= 0, \,\mathrm{and} \\
    \tilde{e}_{i+1} &= 4\beta_{i+1} \bar{\varphi}_i +3\tilde{e}_i. \label{eq:e_tilde_i}
\end{align}
Then $e_i \leq \tilde{e}_i$ for $i \in\{0,\dots, t\}$ and $\varphi_i \leq \pi/4$ for $i \in\{0,\dots, t-1\}$.
\end{claim}

\begin{claim}\label{claim:e_i_tilde}
    $\tilde{e}_i \leq \gamma \bar{\varphi}_i/\pi$ for $\beta_i = \frac{9}{2\pi^3}\gamma/i^2$.
\end{claim}

\begin{proof}
    Let $u_i := \tilde{e}_i/(\gamma \bar{\varphi}_i)$.
    Substituting $\tilde{e}_i = \gamma \bar{\varphi}_i u_i$ to Eq.~\eqref{eq:e_tilde_i}, and using $\bar{\varphi}_{i+1} =3\bar{\varphi}_i$, and finally dividing both sides by $3\gamma \bar{\varphi}_i$, we obtain $u_{i+1} = u_i +\frac{4\beta_{i+1}}{3\gamma}$.
    Since $u_0 =0$, we have:
    \begin{align}
        u_i &= \frac{4}{3\gamma}\sum_{j=1}^i \beta_j
        = \frac{4}{3\gamma}\sum_{j=1}^i \frac{9\gamma}{2\pi^3 j^2}
        = \frac{6}{\pi^3} \sum_{j=1}^i\frac{1}{j^2} \label{eq:u_i_line_1}\\
        &\leq \frac{6}{\pi^3} \frac{\pi^2}{6}
        =\frac{1}{\pi}, \label{eq:u_i_line_2}
    \end{align}
    where we have used the fact that $\sum_{n=1}^{\infty}\frac{1}{n^2} = \frac{\pi^2}{6}$.
    Thus $\tilde{e}_i = \gamma \bar{\varphi}_i u_i \leq \gamma \bar{\varphi}_i/\pi$.
\end{proof}

Combining Claims~\ref{claim:e_i_upper},\ref{claim:e_i_tilde} with $\bar{\varphi}_t \leq \pi/2$, we have $e_t \leq \gamma/2$.
Thus, the final success amplitude $\| \Pi_{M}\otimes I_2^{\otimes t} \cdot A_t \cdot \ket{\pi}\ket{0^t} \| = \sin(\varphi_t)$ is greater than $\sin(\bar{\varphi}_t) -e_t \geq \frac{1}{2} -\frac{\gamma}{2}$, which proves Eq.~\eqref{eq:known_success} in Lemma~\ref{lem:recur_known}.

We now calculate the cost $C(i)$ of $A_i$.
Recall that the cost of the oracle $\mathrm{ref}(\mathcal{M}^\perp)$ is $c_2$, and the cost of the approximate reflection $R(\beta)$ is $c_1 \log(\frac{1}{\beta})$.
From the recursive definition of $A_i$ shown in Algorithm~\ref{alg:recur_known}, we obtain the recurrence relation of $C(i)$ as follows:
\begin{align}
    C(0) &= 0, \ \mathrm{and} \\
    C(i) &= 3C(i-1) +c_1 \log(\frac{1}{\beta_i}) +c_2.
\end{align}
Therefore we can calculate the total cost $C=C(t)$ as follows:
\begin{align}
	C(t) &= \sum_{i=1}^{t} 3^{t-i} \left( c_1 \log(\frac{1}{\beta_i}) +c_2 \right) \label{eq:Ct_line1}\\
	&= 3^t \left[ \Big( c_1 \log(\frac{2\pi^3}{9}) +c_1 \log(\frac{1}{\gamma}) +c_2 \Big) \sum_{i=1}^{t} 3^{-i} +2c_1 \sum_{i=1}^{t} \frac{\log(i)}{3^i} \right] \label{eq:Ct_line2}\\
	&\leq 3^t \left[1.4 c_1 +\frac{c_1}{2} \log(\frac{1}{\gamma}) +c_1 +\frac{c_2}{2}  \right] \label{eq:Ct_line3}\\
	&\leq \frac{\pi}{2 \sqrt{p_M}} \left[c_1 (2.4+\frac{1}{2}\log(\frac{1}{\gamma})) +\frac{c_2}{2} \right] \label{eq:Ct_line4},
\end{align}
which results in Eq.~\eqref{eq:known_cost} in Lemma~\ref{lem:recur_known}.

Above, we expand $\log({1}/{\beta_i}) = \log({2\pi^3 i^2}/{(9\gamma)}) = \log({2\pi^3}/{9}) +\log({1}/{\gamma}) +2\log(i)$ in Eq.~\eqref{eq:Ct_line2}.
Formula~\eqref{eq:Ct_line3} holds by the following three facts: $\sum_{i=1}^{t} 3^{-i}  =\frac{1}{2}(1-1/3^t)$, and $\log(2\pi^3/9)/2 <1.4$, and $\sum_{i=1}^{t} {\log(i)}/{3^i} \leq s:= \sum_{i=1}^{t} {i}/{3^i} = (1-{t}/{3^t})/2$, where the last equality follows from $s -{s}/{3} = {1}/{3} - {t}/{3^{t+1}}$.
Formula~\eqref{eq:Ct_line4} uses $3^t \arcsin(\sqrt{p_M}) \leq \pi/2$ and the fact that $\arcsin(x) \geq x$ for $x\in[0,1]$.

\subsection{Proof of Claim~\ref{claim:e_i_upper}}
Denote by $R_i$ the combined effect of steps 3,4,5 of $A_i$ shown in Algorithm~\ref{alg:recur_known}, which approximates the ideal reflection:
\begin{align}
	\mathrm{ref}(\phi_{i-1}) &:= A_{i-1}(2\ket{\pi,0^{i-1}}\bra{\pi,0^{i-1}}-I) A_{i-1}^\dagger \\
	&= 2\ket{\phi_{i-1}}\bra{\phi_{i-1}}-I.
\end{align}
Denote their difference by $E_i := R_i -\mathrm{ref}(\phi_{i-1}) \otimes I_{\mathcal{K}_i}$, we have the following two facts.

\begin{fact}\label{fact:E_i_phi_i1}
    $E_i \ket{\phi_{i-1},0}=0$.
\end{fact}

\begin{proof}
    It suffices to show that $R_i \ket{\phi_{i-1},0} =\ket{\phi_{i-1},0}$.
    From $\ket{\phi_{i-1}} = A_{i-1} \ket{\pi}\ket{0^{i-1}}$ and $R_i =A_{i-1} \otimes I_{\mathcal{K}_i} \cdot (4.) \cdot A_{i-1}^\dagger \otimes I_{\mathcal{K}_i}$, we have $R_i \ket{\phi_{i-1},0} = A_{i-1}\otimes I_{\mathcal{K}_i} \cdot (4.) \ket{\pi}\ket{0^{i}}$.
    Since $R(\beta_i)$ leaves $\ket{\pi}\ket{0} \in \mathcal{H}\otimes\mathcal{K}_i$ unchanged by Eq.~\eqref{eq:R_0_pi} in Lemma~\ref{lem:approx_ref}, step~4 leaves $\ket{\pi}\ket{0^i}$ unchanged, and thus $R_i \ket{\phi_{i-1},0} = \ket{\phi_{i-1},0}$.
\end{proof}

\begin{fact}\label{fact:E_i_psi}
    Suppose $\ket{\psi} \in \mathcal{H} \otimes \mathcal{K}_1 \otimes \cdots\otimes\mathcal{K}_{i-1}$ satisfies $\ket{\psi} \perp \ket{\phi_{i-1}}$.
    Then $\| E_i \ket{\psi,0} \| \leq \beta_i$.
\end{fact}

\begin{proof}
    Since $\ket{\psi} \perp \ket{\phi_{i-1}}$, we know $\ket{\tilde{\psi}} := A_{i-1}^\dagger \ket{\psi}$ is perpendicular to $A_{i-1}^\dagger \ket{\phi_{i-1}} = \ket{\pi} \ket{0^{i-1}}$.
    It suffices to show that the distance between $(4.)\ket{\tilde{\psi},0}$ and $(2\ket{\pi,0^{i-1}}\bra{\pi,0^{i-1}}-I)\otimes I_{\mathcal{K}_i} \ket{\tilde{\psi},0}=-\ket{\tilde{\psi},0}$ are bounded above by $\beta_i$.
    Decompose $\ket{\tilde{\psi}}$ into two parts conditioned on whether the basis states in $\mathcal{K}_1 \otimes \cdots\otimes\mathcal{K}_{i-1}$ is $\ket{0^{i-1}}$, i.e. $\ket{\tilde{\psi}} =\ket{\tilde{\psi}_0} \ket{0^{i-1}}+\ket{\tilde{\psi}_1}$.
    Denote by $\widetilde{R}(\beta_i)$ the operation that applies $R(\beta_i)$ to $\mathcal{H}\otimes\mathcal{K}_i$ and applies the identity transformation to $\mathcal{K}_1 \otimes \cdots\otimes\mathcal{K}_{i-1}$.
    From the definition of step~4, we have $ (4.) \ket{\tilde{\psi},0} = \widetilde{R}(\beta_i)\ket{\tilde{\psi}_0}\ket{0^{i}} -\ket{\tilde{\psi}_1}\ket{0}$.
    Since $\ket{\tilde{\psi}} \perp \ket{\pi} \ket{0^{i-1}}$, we have $\ket{\tilde{\psi}_0} \perp \ket{\pi}$, and thus $\| ((4.)+I) \ket{\tilde{\psi},0} \| = \| (\widetilde{R}(\beta_i)+I) \ket{\tilde{\psi}_0}\ket{0^{i}} \| \leq \beta_i$, where the last inequality follows from Eq.~\eqref{eq:R_0_psi} in Lemma~\ref{lem:approx_ref}.
\end{proof}

Using Facts~\ref{fact:E_i_phi_i1},\ref{fact:E_i_psi}, we can show the following claim.

\begin{claim}\label{claim:varphi_i1_3}
    $\left|\, \sin(\varphi_{i+1}) - |\sin(3\varphi_i)|\, \right| \leq \beta_{i+1} |\sin(2\varphi_i)|$.
\end{claim}

\begin{proof}
Decompose $\ket{\phi_i} = A_{i}\ket{\pi}\ket{0^i}$, i.e. the state obtained after step~1 in $A_{i+1}$, into two parts using projections $\Pi_M$ and $(I-\Pi_M)$ acting on $\mathcal{H}$, so that $\ket{\phi_i} = \sin(\varphi_i) \ket{\mu_i} +\cos(\varphi_i) \ket{\mu_i^\perp}$.

\begin{figure}[hbt]
	\centering
	\includegraphics[width=0.5\textwidth]{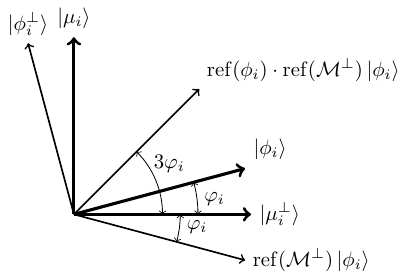}
	\caption{\label{fig:phi_subspace} Decomposing $\ket{\phi_i}$ into two parts using projections $\Pi_M$ and $(I-\Pi_M)$ acting on $\mathcal{H}$, so that $\ket{\phi_i} = \sin(\varphi_i) \ket{\mu_i} +\cos(\varphi_i) \ket{\mu_i^\perp}$.  }
\end{figure}

Let $\ket{\phi_i^\bot} := \cos(\varphi_i)\ket{\mu_i} -\sin(\varphi_i)\ket{\mu_i^\bot}$.
From Fig.~\ref{fig:phi_subspace}, we have:
\begin{equation}\label{eq:later_use_1}
	\mathrm{ref}(\mathcal{M}^\bot) \ket{\phi_i} = \cos(2\varphi_i)\ket{\phi_i} -\sin(2\varphi_i)\ket{\phi_i^\bot},
\end{equation}
which is the state obtained after steps 1 and 2 in $A_{i+1}$.
Technically, $\mathrm{ref}(\mathcal{M}^\perp)$ in \cref{eq:later_use_1} should be $\mathrm{ref}(\mathcal{M}^\perp) \otimes I_{\mathcal{K}_1} \otimes \cdots \otimes I_{\mathcal{K}_i}$, but for simplicity we use the former notation here and after.

From Fig.~\ref{fig:phi_subspace}, we also have:
\begin{equation}\label{eq:later_use_2}
	\mathrm{ref}(\phi_i) \cdot \mathrm{ref}(\mathcal{M}^\perp) \ket{\phi_i} = \sin(3\varphi_i)\ket{\mu_i} +\cos(3\varphi_i) \ket{\mu_i^\bot},
\end{equation}
which is the state we would get after the whole $A_{i+1}$ if steps 3,4 and 5 are replaced with ideal reflection $\mathrm{ref}(\phi_i)$.

From the definition $E_{i+1} = R_{i+1} -\mathrm{ref}(\phi_{i}) \otimes I_{\mathcal{K}_{i+1}}$, we can now write the actual state obtained after the whole $A_{i+1}$, i.e. $\ket{\phi_{i+1}} = R_{i+1}\cdot \mathrm{ref}(\mathcal{M}^\perp) \ket{\phi_i,0}$, as follows:
\begin{align}
	\ket{\phi_{i+1}} &= \mathrm{ref}(\phi_i)\otimes I_{\mathcal{K}_{i+1}} \cdot \mathrm{ref}(\mathcal{M}^\perp) \ket{\phi_i,0} 
	+E_{i+1}\cdot\mathrm{ref}(\mathcal{M}^\perp) \ket{\phi_i,0} \\
	&= \sin(3\varphi_i)\ket{\mu_i,0} +\cos(3\varphi_i) \ket{\mu_i^\bot,0}
	+E_{i+1}(\cos(2\varphi_i)\ket{\phi_i,0} -\sin(2\varphi_i)\ket{\phi_i^\bot,0}) \label{eq:phi_i1_line2}\\
	&= \sin(3\varphi_i)\ket{\mu_i,0} +\cos(3\varphi_i) \ket{\mu_i^\bot,0} +\sin(2\varphi_i) E_{i+1}\ket{\phi_i^\bot,0}, \label{eq:phi_i1_line3}
\end{align}
where we have used Eqs.~\eqref{eq:later_use_1} and \eqref{eq:later_use_2} in Eq.~\eqref{eq:phi_i1_line2}, and Eq.~\eqref{eq:phi_i1_line3} follows from Fact~\ref{fact:E_i_phi_i1}.
Applying $\Pi_M$ to both sides of Eq.~\eqref{eq:phi_i1_line3}, and using the triangle inequality $|\, \|\vec{x}\| -\|\vec{y}\|\, | \leq \|\vec{x} -\vec{y}\|$, we obtain:
\begin{align}
	\left|\, \sin(\varphi_{i+1}) - |\sin(3\varphi_i)|\, \right| 
	&\leq \| \sin(2\varphi_i) \Pi_M E_{i+1}\ket{\phi_i^\bot,0} \| \\
    &\leq \| E_{i+1}\ket{\phi_i^\bot,0} \| \cdot |\sin(2\varphi_i)| \label{eq:claim_3_line2}\\
	&\leq \beta_{i+1} |\sin(2\varphi_i)|, \label{eq:claim_3_line3}
\end{align}
where we have used Fact~\ref{fact:E_i_psi} in Eq.~\eqref{eq:claim_3_line3}.
\end{proof}

We will also need the following two facts.

\begin{fact}\label{fact:sin_2A}
    $|\sin(2A)-\sin(2B)| \leq 2|\sin(A)-\sin(B)|$ for $A,B\in[0,\pi/2]$.
\end{fact}

\begin{proof}
    Write $\sin(2x) = 2\sin(x)\sqrt{1-\sin^2(x)}$ as a function of $\sin(x)$.
    Consider the function $f(s) = s\sqrt{1-s^2}$. Its derivative $f'(s) = \sqrt{1-s^2} - \frac{s^2}{\sqrt{1-s^2}}$ is decreasing for $s\in[0,1]$ and therefore $f'(s) \leq f'(0) =1$.
    Using Lagrange's mean value theorem, we have $|f(s)-f(s')| \leq |s-s'|$ for $s,s'\in[0,1]$, from which Fact~\ref{fact:sin_2A} follows.
\end{proof}

\begin{fact}\label{fact:sin_3A}
    $|\sin(3A)-\sin(3B)| \leq 3|\sin(A)-\sin(B)|$ for $A,B\in[0,\pi/4]$.
\end{fact}

\begin{proof}
    Write $\sin(3x) = 3\sin(x) -4\sin^3(x)$ as a function of $\sin(x)$.
    Consider the function $f(s) = 3s -4s^3$.
    Its derivative $f'(s) = 3-12s^2 \in [-3,3]$ for $s\in[0,1/\sqrt{2}]$.
    Using Lagrange's mean value theorem, we have $|f(s)-f(s')| \leq 3|s-s'|$ for $s\in[0,1/\sqrt{2}]$, from which Fact~\ref{fact:sin_3A} follows.
\end{proof}

We can now show by induction that $e_i \leq \tilde{e}_i$ for $i \in\{0,\dots, t\}$ and $\varphi_i \leq \pi/4$ for $i \in\{0,\dots, t-1\}$.

\noindent\textbf{Base step:} Since $A_0=I$, we have $\sin(\varphi_0) = \|\Pi_M\ket{\pi}\|= \sin(\bar{\varphi}_0)$.
Thus $\varphi_0 = \bar{\varphi}_0 \leq \pi/6$ and $e_0 =0 =\tilde{e}_0$.

\noindent\textbf{Inductive step:} Assuming $e_i \leq \tilde{e}_i$ and $\varphi_i \leq \pi/4$, where $i \leq t-1$.
We first show $e_{i+1} \leq \tilde{e}_{i+1}$ as follows:
\begin{align}
	e_{i+1} &\leq |\sin(\varphi_{i+1}) -\sin(3\varphi_i)| +|\sin(3\varphi_i) -\sin(\bar{\varphi}_{i+1})| \label{eq:e_i1_line1}\\
	&\leq \beta_{i+1}|\sin(2\varphi_i)| +|\sin(3\varphi_i) -\sin(\bar{\varphi}_{i+1})| \label{eq:e_i1_line2}\\
	&\leq \beta_{i+1}(\sin(2\bar{\varphi}_i) +|\sin(2\varphi_i) -\sin(2\bar{\varphi}_i) |) +|\sin(3\varphi_i) -\sin(3\bar{\varphi}_{i})| \label{eq:e_i1_line3}\\
	&\leq \beta_{i+1} (2\bar{\varphi}_i +2e_i) +3e_i \label{eq:e_i1_line4}\\
	&\leq \beta_{i+1} (2\bar{\varphi}_i +2\tilde{e}_i) +3\tilde{e}_i \label{eq:e_i1_line5}\\
	&\leq 4\beta_{i+1} \bar{\varphi}_i +3\tilde{e}_i \label{eq:e_i1_line6}\\
	&= \tilde{e}_{i+1}, \label{eq:e_i1_line7}
\end{align}
where Eq.~\eqref{eq:e_i1_line1} follows from the triangle inequality, and Eq.~\eqref{eq:e_i1_line2} follows from Claim~\ref{claim:varphi_i1_3} and the fact that $|\sin(3\varphi_i)| =\sin(3\varphi_i)$, which holds by the induction hypothesis $\varphi_i \leq \pi/4$.
In Eq.~\eqref{eq:e_i1_line3} we have used the triangle inequality, and $2\bar{\varphi}_i \leq 2\bar{\varphi}_t/3 \leq \pi/3$ so that $|\sin(2\bar{\varphi}_i)| =\sin(2\bar{\varphi}_i)$, and $\bar{\varphi}_{i+1} =3\bar{\varphi}_i$.
Equation~\eqref{eq:e_i1_line4} follows from the fact that $\sin(A)\leq A$ for $A\geq 0$, and Fact~\ref{fact:sin_2A}, Fact~\ref{fact:sin_3A} combined with the induction hypothesis $\varphi_i \leq \pi/4$ and the fact that $\bar{\varphi}_i \leq \pi/6$.
In Eq.~\eqref{eq:e_i1_line5} we have used the induction hypothesis $e_i \leq \tilde{e}_i$.
In Eq.~\eqref{eq:e_i1_line6} we have used $\tilde{e}_i \leq \gamma \bar{\varphi}_i/\pi \leq \bar{\varphi}_i$, which follows from Claim~\ref{claim:e_i_tilde} and $\gamma \in (0,1)$.
Equation~\eqref{eq:e_i1_line7} follows from the recurrence relation of $\tilde{e}_i$ defined by Eq.~\eqref{eq:e_tilde_i}.

We then show that $\varphi_{i+1} \leq \pi/4$ if $i\leq t-2$, using the following inequalities:
\begin{align}
	\sin(\varphi_{i+1}) &\leq \sin(\bar{\varphi}_{i+1}) +\tilde{e}_{i+1} \label{sin_varphi_i1_line1}\\
	&\leq \sin(\bar{\varphi}_{i+1}) +\gamma \bar{\varphi}_{i+1}/\pi \label{sin_varphi_i1_line2}\\
	&\leq \sin(\frac{\pi}{6}) +\frac{1}{6} \label{sin_varphi_i1_line3}\\
	&< \frac{1}{\sqrt{2}} =\sin(\frac{\pi}{4}) \label{sin_varphi_i1_line4},
\end{align}
where Eq.~\eqref{sin_varphi_i1_line1} follows from the definition $e_{i+1} = |\sin(\varphi_{i+1}) -\sin(\bar{\varphi}_{i+1})|$, and $e_{i+1} \leq \tilde{e}_{i+1}$ by Eq.~\eqref{eq:e_i1_line7}.
Equation~\eqref{sin_varphi_i1_line2} follows from Claim~\ref{claim:e_i_tilde}.
Equation~\eqref{sin_varphi_i1_line3} follows from $\bar{\varphi}_{i+1} \leq \bar{\varphi}_{t-1} \leq \pi/6$, since $i\leq t-2$.
The inequality in Eq.~\eqref{sin_varphi_i1_line4} can be verified numerically.

\begin{remark}\label{remark:why_pi_4}
    If we let $t$ be an integer such that $\bar{\varphi}_t \in [\pi/4,3\pi/4]$, which is the case in Ref.~\cite{MNRS}, Eq.~\eqref{sin_varphi_i1_line3} will become $\sin(\pi/4)+\gamma/4$.
    Thus, we cannot guarantee $\varphi_{t-1} \leq \pi/4$, which is needed in Eq.~\eqref{eq:e_i1_line4} to show that $e_t \leq \tilde{e}_t$.
\end{remark}

\section{Proof of Lemma~\ref{lem:recur_unknown}}\label{app:2}
Denote by $\ket{\psi_i} \in \mathcal{H} \otimes \mathcal{K}_1 \otimes \cdots\otimes\mathcal{K}_{i}$ the middle state obtained after Line~7 in the $i$-th loop in Algorithm~\ref{alg:recur_unkonwn}, where the state in $\mathcal{K}_{i+1}, \cdots,\mathcal{K}_{t_\mathrm{max}}$ remains $\ket{0}$.
Denote by $\delta_i := \| \ket{\psi_i} -\ket{\phi_i}\|$ the distance between $\ket{\psi_i}$ and $\ket{\phi_i} = A_i\ket{\pi}\ket{0^i}$.
Let $t$ be the integer such that $3^t\arcsin(\sqrt{p_M}) \in[\pi/2,\pi/6]$. By Lemma~\ref{lem:recur_known} we know $\| \Pi_M \otimes I_2^{\otimes t} \cdot \ket{\phi_t} \| \geq \frac{1}{2}(1-\gamma)$.
We will later prove the following claim.

\begin{claim}\label{claim:delta}
    $\delta_t \leq \frac{\pi}{12} +\frac{3\gamma}{4}$.
\end{claim}

Thus, from Claim~\ref{claim:delta} and the triangle inequality, we can lower bound the success amplitude of $\ket{\psi_t}$ as follows:
\begin{align}
	\| \Pi_M \ket{\psi_t} \| 
	&\geq \| \Pi_M \ket{\phi_t} \| - \| \Pi_M (\ket{\psi_t}-\ket{\phi_t}) \| \\
	&\geq \frac{1}{2}(1-\gamma) -(\frac{\pi}{12} +\frac{3\gamma}{4}) \\
	&= \frac{1}{2} -\frac{\pi}{12}-\frac{5\gamma}{4},
\end{align}
which results in the lower bound on success probability stated in Lemma~\ref{lem:recur_unknown}.

We now calculate the cost $C$ of the search process $S$ shown by Algorithm~\ref{alg:recur_unkonwn} as follows:
\begin{align}
	C &= \sum_{i=1}^{t_\mathrm{max}} (C(i) +c_2) \label{eq:C_line1}\\
	&\leq \sum_{i=1}^{t_\mathrm{max}} 3^i \left[2.4 c_1 +\frac{c_1}{2} \log(\frac{1}{\gamma}) +\frac{c_2}{2}  \right] +c_2 t_\mathrm{max} \label{eq:C_line2}\\
	&\leq \frac{3\pi}{4\sqrt{\varepsilon}} \left[2.4 c_1 +\frac{c_1}{2} \log(\frac{1}{\gamma}) +\frac{c_2}{2}  \right] +c_2 \log(\frac{\pi}{2\sqrt{\varepsilon}}), \label{eq:C_line3}
\end{align}
which results in Eq.~\eqref{eq:unknown_cost} in Lemma~\ref{lem:recur_unknown}.
In Eq.~\eqref{eq:C_line2} we have used the upper bound on the cost $C(i)$ of $A_i$ shown by Eq.~\eqref{eq:Ct_line3}.
In Eq.~\eqref{eq:C_line3} we have used the fact that $\sum_{i=1}^{t} 3^i =3\cdot \frac{3^t-1}{3-1} \leq 3/2\cdot 3^t$, and $3^{t_\mathrm{max}} \leq \pi/(2\arcsin(\sqrt{\varepsilon})) \leq {\pi}/{(2\sqrt{\varepsilon})}$.

\subsection{Proof of Claim~\ref{claim:delta}}

Decompose $\ket{\psi_i} \in \mathcal{H} \otimes \mathcal{K}_1 \otimes \cdots\otimes\mathcal{K}_{i}$, i.e. the middle state obtained after Line~7 in the $i$-th loop in Algorithm~\ref{alg:recur_unkonwn}, into two parts using projections $\Pi_M$ and $(I-\Pi_M)$ acting on $\mathcal{H}$, so that $\ket{\psi_i} = \sin(\theta_i) \ket{v_i} +\cos(\theta_i)\ket{v_i^\perp}$.
From the binary measurement shown by Line~5 in the $(i-1)$-th loop in Algorithm~\ref{alg:recur_unkonwn}, we have $\ket{\psi_i} = A_i \ket{v_{i-1}^\perp,0}$ for $i\geq 1$, and the base case is $\ket{v_0^\bot} = \ket{\pi}$. Recall that $\ket{\phi_i} = A_i\ket{\pi}\ket{0^i} = \sin(\varphi_i) \ket{\mu_i} +\cos(\varphi_i) \ket{\mu_i^\perp}$.
Thus, the distance between $\ket{\psi_{i+1}}$ and $\ket{\phi_{i+1}}$ for $i \in \{1,\dots,t-1\}$ can be bounded above as shown below:
\begin{align}
	\delta_{i+1} &= \| \ket{v_{i}^\bot} -\ket{\pi,0^{i}} \| \label{eq:delta_i1_line1}\\
	&\leq \| \ket{v_{i}^\bot} -\ket{\mu_i^\bot} \| 
	+\sum_{k=1}^{i} \| \ket{\mu_{k}^\bot} -\ket{\mu_{k-1}^\bot,0} \|
	+\| \ket{\mu_0^\bot} -\ket{\pi} \| \label{eq:delta_i1_line2}\\
	&\leq 3\delta_i +4\varphi_0 \sum_{k=1}^{i} 3^{k}\beta_{k} +\varphi_0, \label{eq:delta_i1_line3}
\end{align}
where Eq.~\eqref{eq:delta_i1_line3} uses the Facts~\ref{fact:delta_i1_1}, \ref{fact:delta_i1_2}, \ref{fact:delta_i1_3} shown below.

Using the recurrence relation of $\delta_i$ shown by Eq.~\eqref{eq:delta_i1_line3}, we can write out the first four terms as follows:
\begin{align}
	\delta_1 &= \| \ket{v_0^\bot} -\ket{\pi}\| = 0 \\
	\delta_2 &\leq 4\varphi_0(3\beta_1) +\varphi_0 \\
	\delta_3 &\leq 4\varphi_0[3^2\beta_1 +(3\beta_1 +3^2\beta_2)] +(3+1)\varphi_0 \\
	\delta_4 &\leq 4\varphi_0[3^3\beta_1 +3(3\beta_1 +3^2\beta_2) +(3\beta_1 +3^2\beta_2 +3^3\beta_3)] +(3^2 +3 +1)\varphi_0
\end{align}

It is now clear that
\begin{align}
	\delta_t &\leq 4\varphi_0 \sum_{k=1}^{t-1}\beta_{k} \sum_{j=t-1}^{k} 3^j
    +\varphi_0 \sum_{k=0}^{t-2}3^k \label{eq:delta_t_line1}\\
    &\leq 4\varphi_0 \sum_{k=1}^{t-1}\beta_{k} \cdot \frac{3^t}{2} + \varphi_0 \cdot \frac{3^{t-1}}{2}   \label{eq:delta_t_line2}\\
    &= 2 \bar{\varphi}_t \sum_{k=1}^{t-1}\beta_{k} +  \frac{\bar{\varphi}_t}{6} \label{eq:delta_t_line3}\\
    &\leq 2\cdot\frac{\pi}{2} \cdot \frac{3\gamma}{4\pi} +\frac{\pi}{12} \label{eq:delta_t_line4}\\
    &= \frac{\pi}{12} +\frac{3\gamma}{4}, \label{eq:delta_t_line5}
\end{align}
which proves Claim~\ref{claim:delta}.
We have used $\sum_{j=k}^{t-1} 3^j = 3^k \cdot \frac{3^{t-k}-1}{3-1} < \frac{3^t}{2}$ and $\sum_{k=0}^{t-2}3^k = \frac{3^{t-1}-1}{3-1} < \frac{3^{t-1}}{2}$ in Eq.~\eqref{eq:delta_t_line2}.
Equation~\eqref{eq:delta_t_line3} follows from $\bar{\varphi}_t = 3^t\varphi_0$.
In Eq.~\eqref{eq:delta_t_line4} we have used $\bar{\varphi}_t \leq \pi/2$, and $\frac{4}{3\gamma} \sum_{j=1}^{\infty} \beta_j \leq \frac{1}{\pi}$ which follows from Eqs.~\eqref{eq:u_i_line_1}, \eqref{eq:u_i_line_2}.

\begin{fact}\label{fact:delta_i1_1}
    $\| \ket{\mu_i^\perp} - \ket{v_i^\perp}\| < 3 \delta_i$ for $i \in \{1,\dots,t-1\}$.
\end{fact}
\begin{proof}
Since the distance between $\ket{\psi_i}$ and $\ket{\phi_i} = \sin(\varphi_i) \ket{\mu_i} +\cos(\varphi_i) \ket{\mu_i^\perp}$ is $\delta_i$, we can write $\ket{\psi_i}$ as:
\begin{equation}\label{eq:psi_i_first}
    \ket{\psi_i} = \sin(\varphi_i) \ket{\mu_i} +\cos(\varphi_i) \ket{\mu_i^\perp} +\ket{\xi_i},
\end{equation}
where $\| \ket{\xi_i} \| \leq \delta_i$.
Comparing Eq.~\eqref{eq:psi_i_first} with the other expression of $\ket{\psi_i}$:
\begin{equation}\label{eq:psi_i_second}
    \ket{\psi_i} = \sin(\theta_i) \ket{v_i} +\cos(\theta_i)\ket{v_i^\perp},
\end{equation}
and applying $\Pi_M^\bot$ (technically it should be $(I-\Pi_M)\otimes I_{\mathcal{K}_1} \otimes \cdots \otimes I_{\mathcal{K}_i}$, but we use the former notation for simplicity) to both of Eq.~\eqref{eq:psi_i_first} and Eq.~\eqref{eq:psi_i_second}, we obtain:
\begin{equation}\label{eq:cos_v_u}
	\cos(\theta_i) \ket{v_i^\perp} = \cos(\varphi_i) \ket{\mu_i^\perp} +\Pi_M^\bot \ket{\xi_i}.
\end{equation}
Using the triangle inequality $|\, \|\vec{x}\| -\|\vec{y}\|\, | \leq \|\vec{x} -\vec{y}\|$, we know:
\begin{equation}\label{eq:cos_theta_phi}
	|\cos(\theta_i) -\cos(\varphi_i)| \leq \delta_i.
\end{equation}
From Eq.~\eqref{eq:cos_v_u} we also know:
\begin{equation}\label{eq:cos_mu_v}
	\cos(\varphi_i) (\ket{\mu_i^\perp}-\ket{v_i^\perp}) =(\cos(\theta_i) -\cos(\varphi_i))\ket{v_i^\perp} -\Pi_M^\bot \ket{\xi_i}.
\end{equation}
Combing Eq.~\eqref{eq:cos_theta_phi} with Eq.~\eqref{eq:cos_mu_v}, we have:
\begin{equation}\label{eq:cos_mu_v_delta}
	\cos(\varphi_i) \| \ket{\mu_i^\perp} - \ket{v_i^\perp}\| \leq 2\delta_i.
\end{equation}
From Claim~\ref{claim:e_i_upper} we have $\varphi_i \leq \pi/4$ for $i \in\{0,\dots, t-1\}$, and therefore $\cos(\varphi_i) \leq 1/\sqrt{2}$.
Since $2\sqrt{2} < 3$, we have now proven Fact~\ref{fact:delta_i1_1}.
\end{proof}

\begin{fact}\label{fact:delta_i1_2}
    $\| \ket{\mu_{i+1}^\bot} -\ket{\mu_i^\bot,0}\| < 4\beta_{i+1} 3^{i+1} \varphi_0$ for $i \in \{0,\dots,t-2\}$.
\end{fact}
\begin{proof}
Recall from Eq.~\eqref{eq:phi_i1_line3} that
\begin{align}\label{eq:phi_i1_first}
	\ket{\phi_{i+1}} &= \sin(3\varphi_i)\ket{\mu_i,0} +\cos(3\varphi_i) \ket{\mu_i^\bot,0} +\ket{\omega_{i+1}},
\end{align}
where $\ket{\omega_{i+1}} :=\sin(2\varphi_i) E_{i+1}\ket{\phi_i^\bot,0}$.
Recall from Eqs.~\eqref{eq:claim_3_line2}, \eqref{eq:claim_3_line3} that $\| \ket{\omega_{i+1}} \| \leq \beta_{i+1} |\sin(2\varphi_i)|$, and from Eqs.~\eqref{eq:e_i1_line2}-\eqref{eq:e_i1_line6} that $\beta_{i+1} |\sin(2\varphi_i)| \leq 4\beta_{i+1} \bar{\varphi}_i$ for $i \leq t-1$.
Thus $\| \ket{\omega_{i+1}} \| \leq 4\beta_{i+1} 3^i \varphi_0$, since $\bar{\varphi}_0 =\varphi_0$.
Comparing Eq.~\eqref{eq:phi_i1_first} with the other expression of $\ket{\phi_{i+1}}$:
\begin{equation}\label{eq:phi_i1_second}
    \ket{\phi_{i+1}} = \sin(\varphi_{i+1})\ket{\mu_{i+1}} +\cos(\varphi_{i+1})\ket{\mu_{i+1}^\bot},
\end{equation}
and applying $\Pi_M^\bot$ to both Eq.~\eqref{eq:phi_i1_first} and Eq.~\eqref{eq:phi_i1_second}, we obtain:
\begin{equation}\label{eq:cos_phi_v_u}
	\cos(\varphi_{i+1}) \ket{\mu_{i+1}^\bot} = \cos(3\varphi_{i})\ket{\mu_{i}^\bot,0} +\Pi_M^\bot \ket{\omega_{i+1}}.
\end{equation}
Using the triangle inequality $|\, \|\vec{x}\| -\|\vec{y}\|\, | \leq \|\vec{x} -\vec{y}\|$, we know:
\begin{equation}\label{eq:cos_phi_i_1}
	|\cos(\varphi_{i+1}) -\cos(3\varphi_i)| \leq \| \ket{\omega_{i+1}} \|.
\end{equation}
From Eq.~\eqref{eq:cos_phi_v_u} we also know:
\begin{equation}\label{eq:cos_mu_i_1}
	\cos(\varphi_{i+1})(\ket{\mu_{i+1}^\bot} -\ket{\mu_i^\bot,0}) = (\cos(3\varphi_i) -\cos(\varphi_{i+1}))\ket{\mu_i^\perp} +\Pi_M^\bot \ket{\omega_{i+1}}.
\end{equation}
Combing Eq.~\eqref{eq:cos_phi_i_1} with Eq.~\eqref{eq:cos_mu_i_1}, we have:
\begin{equation}
	\cos(\varphi_{i+1}) \| \ket{\mu_{i+1}^\bot} -\ket{\mu_i^\bot,0}\| \leq  2\cdot 4\beta_{i+1} 3^i \varphi_0.
\end{equation}
From Claim~\ref{claim:e_i_upper} we have $\varphi_i \leq \pi/4$ for $i \in\{0,\dots, t-1\}$, and therefore $\cos(\varphi_{i+1}) \leq 1/\sqrt{2}$ for $i\leq t-2$.
Since $2\sqrt{2} < 3$, we have now proven Fact~\ref{fact:delta_i1_2}.
\end{proof}

\begin{fact}\label{fact:delta_i1_3}
    $\| \ket{\mu_0^\perp} -\ket{\pi} \| \leq \varphi_0$.
\end{fact}
\begin{proof}
    Since $\varphi_0$ is the angle between the unit vectors $\ket{\mu_0^\perp}$ and $\ket{\pi}$ by Fig.~\ref{fig:phi_subspace}, using the geometric fact that the chord length is less than the arc length, we have $\| \ket{\mu_0^\perp} -\ket{\pi} \| \leq \bar{\varphi}_0$.
\end{proof}

\end{document}